\documentclass[11pt]{article}
\usepackage{amsmath,amssymb,amsthm,setspace}
\usepackage{graphicx,psfrag,epsf}
\usepackage{enumerate}
\usepackage[authoryear,longnamesfirst]{natbib}
\usepackage{url} 
\usepackage{amsfonts}
\usepackage{bbm} 
\usepackage{dsfont}
\usepackage{multirow}
\usepackage[table]{xcolor}

\usepackage{xr}
\newcommand{\blind}{1}

\makeatletter

\@addtoreset{equation}{section}
\makeatletter

\newtheorem{theorem}{Theorem}
\newtheorem{corollary}{Corollary}
\newtheorem{lemma}{Lemma}
\newtheorem{proposition}{Proposition}
\newtheorem{assumption}{Assumption}
\theoremstyle{definition}
\newtheorem{remark}{Remark}
\newtheorem{example}{Example}

\newcommand{\norm}[1]{\left \|#1\right \|}
\newcommand{\abs}[1]{\left\vert #1 \right\vert}

\begin{document}
\onehalfspacing

\if1\blind
{
  \title{\bf Estimation and Inference for Moments of Ratios with Robustness against Large Trimming Bias\thanks{First arXiv date: September 4, 2017 (arXiv:1709.00981).  We benefited from useful comments by Peter C. B. Phillips
 (editor), Arthur Lewbel (co-editor), anonymous referees, numerous researchers,  seminar participants at Australian National University, University of Bristol, University of British Columbia, Chinese University of Hong Kong, Duke University, Emory University, Fudan University, Hong Kong University of Science and Technology, University of Melbourne, Monash University, Northwestern University, University of Sydney, University of California Davis, University of California San Diego, University of New South Wales, and University of Technology Sydney, and conference participants at 2018 Asian Meeting of the Econometric Society, 2018 Cemmap Advances in Econometrics, 2018 China Meeting of the Econometric Society, 2019 Asian Meeting of the Econometric Society, 2018 International Association for Applied Econometrics Annual Conference, and New York Camp Econometrics XIII. All remaining errors are ours.}}
  \author{Yuya Sasaki\hspace{.25cm}\\
    Department of Economics, Vanderbilt University\\
    and \\
    Takuya Ura \\
    Department of Economics, University of California, Davis}
		\date{}
  \maketitle
} \fi

\if0\blind
{
  \bigskip
  \bigskip
  \bigskip
  \begin{center}
    {\LARGE\bf Estimation and Inference for Moments of Ratios with Robustness against Large Denominator-Based-Trimming Bias}
\end{center}
  \medskip
} \fi

\bigskip
\maketitle

\begin{abstract}
Researchers often trim observations with small values of the denominator $A$ when they estimate moments of the form $\mathbb{E}[B/A]$. Large trimming is common in practice to reduce variance, but it incurs a large bias. This paper provides a novel method of correcting the large trimming bias. If a researcher is willing to assume that the joint distribution between $A$ and $B$ is smooth, then the trimming bias may be estimated well. Along with the proposed bias correction method, we also develop an inference method. Practical advantages of the proposed method are demonstrated through simulation studies, where the data generating process entails a heavy-tailed distribution of $B/A$. Applying the proposed method to the Compustat database, we analyze the history of external financial dependence of U.S. manufacturing firms for years 2000--2010.
\begin{description}
\item[Keywords] bias correction, ratio, large trimming.
\item[JEL Codes] C13, C14
\end{description}
\end{abstract}

\newpage

\section{Introduction}\label{sec:introduction}

Moments of ratios of the form $\mathbb{E}[B/A]$ are ubiquitous in empirical research. Summary tables in numerous papers report statistics of ratios. The average ratio is sometimes the parameter of interest on its own \citep[e.g.,][]{dunbar/lewbel/pendakur:2017}. In addition, there are research methods that use moments of ratios for identification, e.g., inverse probability weighting \citep{horvitz1952generalization} and special regressor methods \citep{Lewbel:1997,Lewbel:1998,Lewbel:2000}. When some observations have values of the denominator $A$ that are close to zero, they behave as outliers in terms of the ratio, $B/A$, and thus can exercise large influences on the na\"{i}ve sample mean. \citet{khan/tamer:2010} formalize and report irregular asymptotic behaviors of sample moments of $B/A$ in such cases.
\citet{escanciano:2016} provides a discussion of the literature on irregularly identified parameters in a general setup. In the literature of heavy-tailed distributions (e.g., \citealt{pena2008self,peng2017inference}), it is also well known that the  sample mean does not have a Gaussian limit distribution when the distribution is heavy-tailed, i.e., when the tail index is strictly less than 2. When the distribution is heavy-tailed, the null distribution for $t$-statistics can be bimodal at $\pm 1$ \citep{fiorio2010bimodal}.\footnote{In the probability literature, a related result is found in \cite{logan1973limit}.}

To alleviate this outlier problem, practitioners often trim observations with small $A$ or large $B/A$. However, a trimmed mean can induce a non-negligible bias in the limit distribution. In fact, {as \citet[][Theorem 3.1.c]{chaudhuri/hill:2016} point out, if $B/A$ has a heavy tail in the sense that its tail index is strictly less than 2, then a trimmed mean without bias correction would not entail the asymptotic normality about $\mathbb{E}[B/A]$ even based on self-normalized sums.}\footnote{\citet{Yang/Ding:2018} derive the asymptotic normality for a trimmed mean without bias correction, which is effective for a different parameter from $\mathbb{E}[B/A]$.} \cite{yang:2014} and \cite{chaudhuri/hill:2016}, therefore, exploit restrictions on the tail behavior of $B/A$ to correct the biases in their proposed estimators, and develop methods of inference based on their bias-corrected trimmed estimators that are asymptotically valid even under heavy-tailed distributions.

In this paper, we propose a novel method of estimation and inference for $\mathbb{E}[B/A]$. Our proposed method is based on a bias-corrected trimmed estimator of the form {$\widehat\theta(h_n) = \widetilde\theta(h_n) - \widehat\lambda(h_n)$, where $h_n$ is  a trimming threshold,  $\widetilde\theta(h_n) = n^{-1}\sum_{i=1}^n ({B_i}/{A_i}) \cdot \mathbbm{1}\left\{A_i \ge h_n\right\}$ is a trimmed sample mean estimator and $\widehat\lambda(h_n)$ is a bias estimator.} In a similar way to \cite{yang:2014} and \cite{chaudhuri/hill:2016}, trimming and bias correction together allow for valid inference about $\mathbb{E}[B/A]$ even under heavy-tailed distributions.
Unlike these papers, on the other hand, we do not need to rely on restrictions on the tail behavior of the distribution of   $B/A$. Specifically, while \cite{chaudhuri/hill:2016} assume regularly varying tails with tail index strictly greater than 1, we do not need such restrictions on the tail behavior. Instead, we propose to exploit smoothness of the conditional expectation function $a \mapsto \mathbb{E}[B|A=a]$ to derive more practical and theoretical benefits compared to the existing alternative methods.

The main advantage of our proposed method is that a greater smoothness of the conditional expectation function $a \mapsto \mathbb{E}[B|A=a]$ allows for a wider admissible range of the trimming threshold $h_n$. This wider range contributes to insensitivity of valid inference to the exact trimming level. In addition, a larger trimming threshold $h_n$ allowed by the wider range can yield faster convergence rates of the trimmed estimator $\widetilde\theta(h_n)$ and thus of our bias-corrected estimator $\widehat\theta(h_n)$ as well. In this sense, we propose to exploit the smoothness of $a \mapsto \mathbb{E}[B|A=a]$ to achieve more insensitive inference and faster convergence rate of the bias-corrected trimmed estimator. The latter is analogous to the well-established idea in the kernel methods where greater smoothness of the underlying function is exploited to achieve faster convergence rates via higher order kernels and higher order local polynomials.

\cite{yang:2014} and \cite{chaudhuri/hill:2016} consider the case of trimming tail observations in terms of the fraction $B/A$.\footnote{\label{footnote:trimming}\citet[][Appendix G]{chaudhuri/hill:2016b} also mention the case of denominator trimming. {Considering the case in which $A=A(X)$ and $B=B(X)$ depend on covariates, \citet{khan/tamer:2010} propose to trim observations based on $X$.}} In contrast, we consider the case of trimming close-to-zero observations in terms of the denominator $A$.\footnote{\citet[][pg. 2125]{graham2012identification} informally suggest a similar approach in the context of correlated random coefficient panel models.} The twin motives for our taking this approach to trimming based on $A$ as opposed to $B/A$ are: First, we aim to provide means to improve the estimation method used commonly by researchers who, often motivated by \cite{crump2009dealing} and others, trim close-to-zero observations in terms of the denominator $A$. Second, more importantly, while the bias from trimming based on $B/A$ can be approximated based on tail shapes as in \cite{yang:2014} and \cite{chaudhuri/hill:2016}, the bias from trimming based on the denominator $A$ can be approximated based on the conditional expectation function $a \mapsto \mathbb{E}[B|A=a]$. This is why, as emphasized above, we do not need to rely on restrictions on tail behaviors of the heavy-tailed distribution. Furthermore, the smoothness of this function determines the extent to which the trimming bias can be corrected. Greater smoothness allows for better bias correction, which in turn allows for larger trimming and hence smaller variances or faster rates of convergence. This advantage is made possible by trimming based on the denominator $A$ rather than on the fraction $B/A$.

In order to provide a complete empirical procedure, we develop an inference method for $\mathbb{E}[B/A]$ based on $\widehat\theta(h_n)$.\footnote{To this end, we take advantage of asymptotic distribution theories for trimmed sums and self-normalized sums. For the former, see the early literature, e.g.,  \cite{griffin1987central,csorgo1988asymptotic,griffin1989asymptotic} and references therein. For the latter, see \cite{andrews1998semiparametric,Romano1999,peng:2001,peng2004empirical,pena2008self,Antoine/Renault:2009,hill/renault:2010,Antoine/Renault:2012,chen2014sieve,chen2015sieve} in addition to the references in the main text that are more closely related to this paper in terms of the ratio structure.} There is an extensive body of literature on inference in a related framework \citep[e.g.,][]{khan/tamer:2010,yang:2014,khan/nekipelov:2015,chaudhuri/hill:2016,rothe:2017,ma/wang:2018,heiler2019valid,hong2020inference}. All these papers use distributional restrictions in different ways. \citet{khan/tamer:2010} assume that Lindeberg's condition holds with negligible bias, implying that the tail index is 2 or above. \citet{yang:2014} considers tail shapes and trimming choice leading to an $n^{-1}$-rate of the leading bias from trimming. \citet{khan/nekipelov:2015} focus on cases where the tail index is local to 2. \citet{chaudhuri/hill:2016}, \citet{ma/wang:2018}, and \citet{heiler2019valid} assume regularly varying tails with tail index greater than 1.
\citet{rothe:2017} takes advantage of a normal distribution assumption. \citet{hong2020inference} consider a finite support with drifting sequence. We do not make any of these distributional assumptions. As emphasized earlier, we instead rely on the smoothness of the conditional moment function $a \mapsto \mathbb{E}[B|A=a]$.

Of these related papers, the most closely related are the pioneers, \citet{yang:2014} and \citet{chaudhuri/hill:2016}, in proposing bias-corrected trimmed estimation for asymptotically valid inference under heavy-tailed distributions, as emphasized above. It seems that this paper is the first to propose an asymptotically valid inference for $\mathbb{E}[B/A]$ under heavy-tailed distributions without a restriction on tail behaviors.\footnote{The first arXiv date of the present paper is September 4, 2017 (arXiv:1709.00981).} More recently, \citet{ma/wang:2018} apply this novel approach of denominator-based-trimmed estimation and propose an important bridge between the non-standard limit distributions of the mean estimators with small or no trimming \citep[as in][]{heiler2019valid} and the Gaussian limit distribution of the bias-corrected denominator-based-trimmed mean estimator (as in our result). All these preceding papers impose an assumption of regularly varying tail behaviors, unlike this paper.

{\bf Notations:} $\mathbb{E}[X]$ and $Var(X)$ denote the expected value and the variance of random variable $X$, respectively. The sample mean is denoted by $\mathbb{E}_n[X]=n^{-1}\sum_{i=1}^n X_i$. The convergence in distribution is denoted by $\stackrel{d}{\rightarrow}$. The indicator function is denoted by $\mathbbm{1}\{\cdot\}$.

{\bf Outline of the paper:} Section \ref{sec:overview} presents an overview of our proposed method without theoretical details. Section \ref{sec:main_results} presents supporting theories. Section \ref{sec:discussion} discusses the assumptions in terms of concrete structures and concrete estimators. Section \ref{sec:properties} presents three important properties of the estimator. Section \ref{sec:implementation} presents a practical guideline. Section \ref{sec:extension} contains extended results to the case of estimated or generated $A$ and $B$. Sections \ref{sec:simulation} and \ref{sec:empirical} present simulation studies and an empirical illustration, respectively. Section \ref{sec:conclusion} concludes. Appendices \ref{sec:proofs_of_the_main_results}--\ref{sec:auxiliary_lemmas_for_the_extended_results} contain mathematical proofs.

\section{Overview of the Proposed Method}\label{sec:overview}

The current section presents an overview of our proposed method without detailed discussions of the theories supporting it. This overview is aimed to serve as a concise guideline. In Section \ref{sec:main_results}, we formally present a theoretical rationale for the proposed methodology.

Suppose that a researcher is interested in estimation and inference for the moment of ratios of the form $\mathbb{E}\left[{B}/{A}\right]$ where we observe $n$ independent copies of $(A,B)$ in a dataset. Throughout this paper, we define the parameter of interest more generally by
\begin{align}\label{eq:definition_parameter_of_interest}
\theta_0 = \lim_{h\downarrow 0}\mathbb{E}\left[\frac{B}{A}\cdot \mathbbm{1}\{\abs{A} \geq h\}\right].
\end{align}
Note that: (i) $\theta_0=\mathbb{E}\left[{B}/{A}\right]$ if $\abs{A}>0$ almost surely and $\mathbb{E}\left[|{B}/{A}|\right]<\infty$; and (ii) $\theta_0$ can be well-defined even if $A$ has a point mass at $0$.\footnote{For example, assume that $\frac{B}{A}\cdot \mathbbm{1}\{\abs{A} \geq h\}$ is integrable for every $h>0$ and that $m(a)=\mathbb{E}[B\mid A=a]$ satisfies $|m(a)|\leq C|a|$ for a finite constant $C$. In this case, even if $A$ has a point mass at $0$, we have 
$$
\abs{\theta_0}
= 
\abs{\lim_{h\downarrow 0}\mathbb{E}\left[\frac{B}{A}\cdot \mathbbm{1}\{\abs{A} \geq h\}\right]}
\leq
\lim_{h\downarrow 0}\mathbb{E}\left[\frac{\abs{m(A)}}{\abs{A}}\cdot \mathbbm{1}\{\abs{A} \geq h\}\right]
= 
C\lim_{h\downarrow 0}\mathbb{E}\left[\mathbbm{1}\{\abs{A} \geq h\}\right],
$$
and therefore $\theta_0$ is finite.} 
In this paper, we allow the tail index for $B/A$ to be less than 1, so that the mean of $B/A$ may not be well-defined but (\ref{eq:definition_parameter_of_interest}) is; in fact, the tail index is not well-defined when A has a point mass. The following two examples illustrate cases where $\mathbb{E}\left[B/A\right]$ is of interest in economic research.

\begin{example}[Inverse Probability Weighting]\label{ex:ipsw}
Let $Y = (1-D)Y_0 + DY_1$ denote an observed outcome, where $D$ is an observed binary indicator of treatment, $Y_0$ is an unobserved potential outcome under no treatment, and $Y_1$ is an unobserved potential outcome under treatment. With the knowledge of the propensity score $P$, a researcher can identify the mean potential outcome $\mathbb{E}[Y_1]$ under treatment by the moment $\mathbb{E}\left[B/A\right]$ where $B=DY$ and $A=P$. (The propensity score $A=P$ is often generated by probit or logit in practice. We also present an extension of our theory to such cases in Section \ref{sec:extension}.) $\triangle$ 
\end{example}

\begin{example}[Binary Choice Model]\label{ex:special_regressor}
Consider the binary choice model of \citet{Lewbel:1997} where $Y = \mathbbm{1}\{\alpha + V - U \geq 0\}$, $V$ and $U$ are independent, both $V$ and $U$ have support on $\mathbb{R}$, and $\mathbb{E}[U]=0$. \citet{Lewbel:1997} shows that $\alpha$ is identified by the moment $\mathbb{E}\left[B/A\right]$ where $B=Y-\mathbbm{1}\{V>0\}$ and $A=f_V(V)$. This technique has been extended to a wide range of econometric models -- see \citet{lewbel:2014}. Also see \citet{khan/tamer:2010} in its motivation related to the present paper.
$\triangle$ 
\end{example}

When there exist observations with small denominator $A$, the na\"ive sample mean estimator $\mathbb{E}_n\left[B/A\right]$ of $\theta_0$ may entail a large variance. To deal with this issue by following a common practice in empirical research, consider the denominator-based-trimmed mean estimator of the form
\begin{align}\label{eq:trimmed_estimator}
\widetilde\theta(h_n) = \mathbb{E}_n\left[\frac{B}{A} \cdot \mathbbm{1}\{A \geq h_n\}\right],
\end{align}
where $h_n>0$ denotes a trimming threshold. Here, we consider the case where $A$ is non-negative -- an extension to more general cases is omitted since relevant applications involve only non-negative $A$ -- see Examples \ref{ex:ipsw} and \ref{ex:special_regressor} above. We further normalize the support of $A$ to $[0,1]$ for simplicity.

Trimming reduces the variance of an estimator on the one hand, but it exacerbates the bias $\mathbb{E}\left[\widetilde\theta(h_n)\right] - \theta_0$ on the other hand. To correct the bias from the trimming in (\ref{eq:trimmed_estimator}), we use the information about the joint distribution between $A$ and $B$ as emphasized in the introduction. Specifically, assuming that the conditional expectation function $m(\cdot) = \mathbb{E}\left[B \ | \ A= \ \cdot \ \right]$ is $k$-times continuously differentiable  at $a=0$, we estimate the $k-1$ derivatives $\left(m^{(1)}(0),...,m^{(k-1)}(0)\right)'$ of $m$ by
$$
\widehat m^{(\kappa)}(0) = p_K^{(\kappa)}(0)' \mathbb{E}_n\left[ p_K(A)p_K(A)' \right]^{-1} \mathbb{E}_n\left[ p_K(A)B \right]
$$
for each $\kappa \in \{1,...,k-1\}$, where $p_K(a)$ denotes the $(K+1)$-dimensional vector of the shifted orthonormal Legendre polynomial basis of degree $K \geq k$:\footnote{While we present the complete basis here including the constant term, we can exclude the constant term in our proposed method due to Assumption \ref{a:m} (i) in Section \ref{sec:main_results}.}
\begin{align}\label{eq:legendre}
p_K(a) = &\left(\begin{array}{c} 
1  \\ 
\sqrt{3}(2a-1)\\ \sqrt{5}(6a^2-6a+1)\\ \sqrt{7}(20a^3-30a^2+12a-1)\\ \sqrt{9}(70a^4-140a^3+90a^2-20a+1)\\ \sqrt{11}(252a^5-630a^4+560a^3-210a^2+30a-1)\\ \vdots\end{array}\right).
\end{align}
With these derivative estimates $\left(\widehat m^{(1)}(0),...,\widehat m^{(k-1)}(0)\right)'$, we propose that
\begin{align}\label{eq:bias_estimator}
\widehat \lambda(h_n) = - \sum_{\kappa=1}^{k-1} \frac{\mathbb{E}_n\left[ A^{\kappa-1} \cdot \mathbbm{1}\{ 0<A<h_n \}\right]}{\kappa!} \cdot \widehat m^{(\kappa)}(0)
\end{align}
estimates the trimming bias $\mathbb{E}\left[\widetilde\theta(h_n)\right] - \theta_0$.
In other words, our proposed bias-corrected trimmed mean estimator is
\begin{align}\label{eq:estimator}
\widehat\theta(h_n)
\equiv
\widetilde\theta(h_n) - \widehat \lambda(h_n)
=
\mathbb{E}_n\left[\frac{B}{A} \cdot \mathbbm{1}\{A \geq h_n\}\right] + \sum_{\kappa=1}^{k-1} \frac{\mathbb{E}_n\left[ A^{\kappa-1} \cdot \mathbbm{1}\{ 0<A<h_n \}\right]}{\kappa!} \cdot \widehat m^{(\kappa)}(0).
\end{align}

The standard error of the bias-corrected trimmed mean estimator (\ref{eq:estimator}) can be estimated by
$$
n^{-1/2} \cdot 
\left(
\mathbb{E}_n\left[\left( \frac{B}{A} \cdot \mathbbm\{ A \geq h_n \} + \widehat{c}(h_n)' \widehat \psi \right)^2\right]
-\mathbb{E}_n\left[\frac{B}{A} \cdot \mathbbm\{ A \geq h_n \} + \widehat{c}(h_n)' \widehat \psi\right]^2
\right)
^{1/2},
$$
where $\widehat{c}(h_n)$ is the $(k-1)$-dimensional vector defined by
$$
\widehat{c}(h_n) = 
\left(\begin{array}{c}
\mathbb{E}_n\left[\mathbbm{1}\{0<A<h_n\}\right]/1!\\
\mathbb{E}_n\left[A \cdot \mathbbm{1}\{0<A<h_n\}\right]/2!\\
\vdots\\ 
\mathbb{E}_n\left[A^{k-2} \cdot \mathbbm{1}\{0<A<h_n\}\right]/(k-1)!\\
\end{array}\right),
$$
and $\widehat \psi = \left(\widehat\psi_1,...,\widehat\psi_{k-1}\right)'$ is the $(k-1)$-dimensional vector defined by
$$
\widehat\psi_\kappa = p_K^{(\kappa)}(0)' p_K(A) \left(B - p_K(A)' \widehat\beta\right)
\qquad\text{with }
\widehat\beta = \mathbb{E}_n\left[ p_K(A)p_K(A)' \right]^{-1} \mathbb{E}_n\left[ p_K(A)B \right]
$$
for each $\kappa \in \{1,...,k-1\}$.

\begin{remark}\label{remark:local}
We propose to use the sieve estimation for bias correction. A natural question is whether one can use local polynomial estimators instead. If the density of $A$ is large near 0, then we can certainly estimate the derivative of $m(\cdot)$ using a local polynomial regression. On the other hand, if the density of $A$ is small near 0, then a local polynomial estimation of the derivative will incur a large variance because not many local observations are available. Since we develop a method of rate-adaptive inference in this paper, we would like our framework to cover both of these two alternative cases. Due to the suboptimal performance of the local polynomial regression in the latter case, we use the global series approximation.
\end{remark}

\section{Main Results}\label{sec:main_results}

For a short-hand notation, we write the population counterpart of the trimmed mean estimator (\ref{eq:trimmed_estimator}) by $\theta(h_n) = \mathbb{E}\left[\widetilde\theta(h_n)\right]$. With this notation, the bias of the trimmed mean estimator $\widetilde\theta(h_n)$ can be simply written as $\theta(h_n)-\theta_0$. In order to characterize this trimming bias, we use the information about the joint distribution between the numerator $B$ and the denominator $A$ as emphasized in Section \ref{sec:introduction}. Specifically, we use the conditional expectation function $m(\cdot) = \mathbb{E}\left[B \ | \ A= \ \cdot \ \right]$ of the numerator $B$ given the denominator $A$. With this notation, the joint distribution of $(A,B)$ is supposed to meet the following assumption.

\begin{assumption}\label{a:m}
(i) $m(0) = 0$.
(ii) $m$ is $k$-times continuously differentiable with a bounded $k$-th derivative in a neighborhood of 0.
(iii) $A\in[0,1]$.
(iv) $\mathbb{E}\left[B^4\right] < \infty$.
(v) $\theta_0$ is well-defined and finite. 
\end{assumption}
 
We discuss this and other assumptions in detail in Section \ref{sec:discussion}, but it is worth mentioning here that the majority of Assumption \ref{a:m} are fairly mild requirements. More importantly, we emphasize that the restriction on the joint distribution of $(A,B)$ imposed by Assumption \ref{a:m} is compatible with cases where the distribution of $B/A$ has a heavy tail due to small $A$ and trimming matters to improve the convergence rate -- see Section \ref{sec:heavy_tail} for details. In particular, Assumption \ref{a:m} allows for infinite variance of $B/A$. We stress that $k$, which appears in part (ii), is \textit{not} a tuning parameter. In the related literature, distributional assumptions are often imposed, but different papers do so in different ways. For instance, \citet{rothe:2017} assumes the normal distribution assumption. \citet{chaudhuri/hill:2016} assume that the distribution of $B/A$ is regularly varying. On the other hand, we assume a known degree of smoothness for $\mathbb{E}[B|A]$.

It turns out that we can approximate the bias $\theta(h_n)-\theta_0$ of the trimmed mean estimator $\widetilde\theta(h_n)$ up to the order $k$ of smoothness given in Assumption \ref{a:m} (ii). This is the sense in which we mean by our statement that the bias can be approximated by exploiting the information about the joint distribution between the numerator $B$ and the denominator $A$ as emphasized in Section \ref{sec:introduction}. The following theorem provides a bias characterization by using the conditional expectation function, $m(\cdot)=\mathbb{E}\left[B \ | \ A= \ \cdot \ \right]$, of the numerator $B$ given the denominator $A$.

\begin{theorem}[Bias Characterization]\label{lemma:bias}
If Assumption \ref{a:m} (i)--(ii) is satisfied, then\small
\begin{align*}
\theta_0 - \theta(h_n) 
&= 
\sum_{\kappa=1}^{k-1} \frac{\mathbb{E}\left[A^{\kappa-1} \cdot \mathbbm{1}\{0<A<h_n\}\right]}{\kappa!} \cdot m^{(\kappa)}(0)
\\&\qquad 
+
\frac{\mathbb{E}\left[ A^{k-1} \cdot \mathbbm{1}\{0<A<h_n\} \cdot \int_0^1 (1-t)^{k-1} m^{(k)}(tA)dt \right]}{(k-1)!}.
\end{align*}
\end{theorem}
 
A proof of this theorem is given in Appendix \ref{sec:lemma:bias}. The first $k-1$ terms of this bias characterization in Theorem \ref{lemma:bias} motivate the bias estimator $\widehat \lambda(h_n)$ provided in (\ref{eq:bias_estimator}), and thus the bias-corrected trimmed mean estimator $\widehat\theta(h_n)$ provided in (\ref{eq:estimator}). 

We next obtain the asymptotic distribution of the bias-corrected trimmed mean estimator $\widehat\theta(h_n)$. To this end, the trimming threshold, $h_n$, of the trimmed mean estimator is chosen to satisfy the following assumption, detailed discussions of which are given in Section \ref{sec:a:h}.

\begin{assumption}\label{a:h}
(i) $nh_n^{2(k-1)} = O(1)$; and
(ii) $n^{-1}h_n^{-4} = o(1)$
as $n \rightarrow \infty$.
\end{assumption}
 
The main role of Assumption \ref{a:h} (i) is to ensure that the remaining bias is negligible. Assumption \ref{a:h} (ii) will be used to establish Lyapunov's condition.

For convenience of writing asymptotic linear representations for the bias-corrected trimmed mean estimator, we introduce the notation
\begin{align}\label{eq:z}
Z(h_n) = \frac{B}{A} \cdot \mathbbm{1}\{A \geq h_n\} + c(h_n)' \psi,
\end{align}
where $c(h_n)$ is the $(k-1)$-dimensional vector defined by
\begin{align}\label{eq:c}
c(h_n) = 
\left(\begin{array}{c}
\mathbb{E}\left[\mathbbm{1}\{0<A<h_n\}\right]/1!\\
\mathbb{E}\left[A \cdot \mathbbm{1}\{0<A<h_n\}\right]/2!\\
\vdots\\
\mathbb{E}\left[A^{k-2} \cdot \mathbbm{1}\{0<A<h_n\}\right]/(k-1)!\\
\end{array}\right),
\end{align}
and $\psi$ is the influence function of an estimator $\left(\widehat m^{(1)}(0),...,\widehat m^{(k-1)}(0)\right)$ of the $k-1$ derivatives $\left( m^{(1)}(0),..., m^{(k-1)}(0)\right)$ of the conditional expectation function $m$ at 0. In Sections \ref{sec:overview} and \ref{sec:discussion}, we present concrete instances of $\left(\widehat m^{(1)}(0),...,\widehat m^{(k-1)}(0)\right)$ and $\psi$ that satisfy theoretical requirements. In the current section on general theory, however, they are assumed to satisfy the following general conditions.

\begin{assumption}\label{a:psi}
(i)
$
\left(1,...,h_n^{k-2}\right)' \circ
\left(
\widehat m^{(1)}(0) - m^{(1)}(0), ...,
\widehat m^{(k-1)}(0) - m^{(k-1)}(0)
\right)'
=$ \\ $\left(1,...,h_n^{k-2}\right)' \circ \mathbb{E}_n\left[\psi\right] + o_p\left(n^{-1/2}\right)$;
(ii)
$\mathbb{E}\left[\psi\right]=0$;
and
(iii) $\mathbb{E}\left[\left(\left(1,...,h_n^{k-2}\right) \psi\right)^4\right]^{1/4} = O\left(n^{1/4}\right)$.
\end{assumption}
 
The symbol `$\circ$' used in Assumption \ref{a:psi} (i) denotes the Hadamard product of vectors. Instead of directly stating a concrete, primitive, and sufficient condition, we state this high-level condition for the benefit of clearly highlighting how each of the components, (i), (ii), and (iii), is used for key steps in our proof. As remarked earlier, we discuss this and other assumptions in Section \ref{sec:discussion} in detail, and provide a concrete, primitive, and sufficient  condition for Assumption \ref{a:psi} in Section \ref{sec:a:psi}.

Together with the bias characterization of Theorem \ref{lemma:bias}, Assumption \ref{a:psi} is used to derive the asymptotic linear representation
$$
\widehat\theta(h_n) - \theta_0 = (\mathbb{E}_n-\mathbb{E})\left[Z(h_n)\right] + o_p(n^{-1/2}),
$$
where the formal statement is given in Lemma \ref{lemma:linear_representation} of Appendix \ref{sec:lemma:linear_representation}. Given this asymptotic linear representation, it now remains to check the conditions for a self-normalized central limit theorem for the triangular array of $Z(h_n)$. In particular, Lyapunov's condition for the triangular array of $Z(h_n)$ can be established in Lemma \ref{lemma:l4l2} of Appendix \ref{sec:lemma:l4l2} under Assumptions \ref{a:m} (iv), \ref{a:h} (ii), and \ref{a:psi} (iii). Consequently, applying Lyapunov's central limit theorem, we obtain the following self-normalized asymptotic normality for the bias-corrected trimmed mean estimator $\widehat\theta(h_n)$. 
Note that $\left(\widehat m^{(1)}(0),...,\widehat m^{(k-1)}(0)\right)$ estimates higher-order derivatives of the conditional expectation function $m(\cdot)$, and therefore they may have a high sampling variability. 
The asymptotic linear representation takes care of this sample variability since $Z(h_n)$ includes the influence function, $\psi$, for $\left(\widehat m^{(1)}(0),...,\widehat m^{(k-1)}(0)\right)$. 
A proof of Theorem \ref{theorem:asymptotic_distribution} is provided in Appendix \ref{sec:theorem:asymptotic_distribution}.

\begin{theorem}[Asymptotic Distribution]\label{theorem:asymptotic_distribution}
If Assumptions \ref{a:m}, \ref{a:h}, and \ref{a:psi} are satisfied, then
\begin{align*}
\frac{ \widehat\theta(h_n) - \theta_0 }{\sqrt{\sigma(h_n)^2/n}}
\stackrel{d}{\rightarrow} \mathcal{N}(0,1),
\end{align*}
provided that $\sigma(h_n)^2\equiv Var(Z(h_n))$ is bounded away from zero.
\end{theorem}

It follows from Assumption \ref{a:h} that this limit distribution result accommodates the rate range of $n^{-1/4} \prec h_n \precsim n^{-1/(2k-1)}$ for a choice of the trimming threshold $h_n$, and this range expands as the smoothness index $k$ increases. In other words, the extent of insensitivity of this limit distribution increases with the degree of smoothness. This is a practical advantage of exploiting the smoothness, in addition to the theoretical advantage of allowing for faster convergence rates -- see Section \ref{sec:faster}.

\section{Discussions of the Assumptions}\label{sec:discussion}

In Section \ref{sec:main_results}, we presented a general theory using three high-level assumptions (Assumptions \ref{a:m}, \ref{a:h}, and \ref{a:psi}). These general conditions accommodate a number of alternative structures and alternative estimators. In the current section, we complement the general theory by discussing Assumptions \ref{a:m}, \ref{a:h}, and \ref{a:psi} with concrete structures and concrete instances. The guideline presented in Section \ref{sec:overview} is also based on these concrete instances as well as the general theory of Section \ref{sec:main_results}.

 \subsection{On Assumption \ref{a:m}}

This assumption consists of four requirements on the conditional expectation function $m$ of the numerator $B$ given the denominator $A$. The four parts are discussed in Sections \ref{sec:part1}, \ref{sec:part2}, \ref{sec:part3}, and \ref{sec:part4} below. Some of the assumptions are discussed in terms of Example \ref{ex:ipsw} and Example \ref{ex:special_regressor} introduced in Section \ref{sec:overview}. Furthermore, we also discuss the compatibility of this assumption with heavy-tailed distributions of $B/A$ in Section \ref{sec:heavy_tail}.

\subsubsection{On Assumption \ref{a:m} (i): $m(0) = 0$}\label{sec:part1}
This assumption is an innocuous assumption in Example \ref{ex:ipsw}. Since $D=0$ holds almost surely given $P=0$, we have $B=0$ almost surely given $A=0$. Therefore, $m(0)=\mathbb{E}[B|A=0]=\mathbb{E}[0|A=0]=0$, and thus this assumption is automatically satisfied. For Example \ref{ex:special_regressor}, this assumption is satisfied if $\inf_{v: f_V(v)=a}(F_{U}(\alpha+|v|)-F_{U}(\alpha-|v|))\rightarrow 1\mbox{ as }a\rightarrow 0$. See Assumption B.3 in \cite{Lewbel:1998}. 

\subsubsection{On Assumption \ref{a:m} (ii): $m$ is $k$-times continuously differentiable with a bounded $k$-th derivative in a neighborhood of 0.}\label{sec:part2}
If a researcher is willing to accept this assumption with a higher order $k$, then the trimming bias can be corrected to a larger extent.
Such a relation is characterized by Theorem \ref{lemma:bias} in Section \ref{sec:main_results}.
This concept is analogous to the related concept in general nonparametric methods where an assumption of higher-order smoothness of a function allows for higher-order bias reduction, e.g., via higher-order kernel or higher-order local polynomials.

In Example \ref{ex:ipsw}, under the unconfoundedness condition \citep{Rosenbaum1983}, this assumption is true if the conditional expectation $\mathbb{E}\left[ Y_1 | P = p \right]$ of the potential outcome $Y_1$ under treatment given the propensity score $P=p$ is a smooth function of $p$.
For Example \ref{ex:special_regressor}, this assumption is true if regular conditional probabilities $P(V>0|f_V(V)= \cdot )$ and $P(\alpha+V-U\geq 0|f_V(V)= \cdot )$ exist and are smooth around zero.

\subsubsection{On Assumption \ref{a:m} (iii): $A\in[0,1]$.}\label{sec:part3}
This condition is made to simplify the proof, and it is not essential in this paper. 

\subsubsection{On Assumption \ref{a:m} (iv): $\mathbb{E}\left[B^4\right] < \infty$.}\label{sec:part4}
This assumption requires that the distribution of the numerator $B$ has a bounded fourth moment.
Note that we allow for infinite fourth moment, and even infinite variance, for distributions of the fraction $B/A$ due to small values of the denominator $A$.
However, we rule out infinite fourth moment of the numerator \textit{per se}.

In Example \ref{ex:ipsw}, a bounded fourth moment of the potential outcome $Y_1$ under treatment is sufficient.
For Example \ref{ex:special_regressor}, this assumption is satisfied since $B$ is bounded.

\subsubsection{On Assumption \ref{a:m} (v): $\theta_0$ is well-defined and finite.}\label{sec:part5}
Our theory requires that the parameter of interest needs to be well-defined and finite.

\subsubsection{Compatibility of Assumption \ref{a:m} with Heavy-Tailed Distributions}\label{sec:heavy_tail}
As the final remark on Assumption \ref{a:m}, we discuss the compatibility of Assumption \ref{a:m} with heavy-tailed distributions of $B/A$.
Here we make this argument to demonstrate that Assumption \ref{a:m} is sufficiently mild to accommodate data generating processes where trimming matters.
Trimming is useful particularly when the distribution of $B/A$ has a heavy tail due to small $A$.
Specifically, trimming matters non-trivially when the variance $Var(B/A)$ does not. 
Furthermore, large trimming (i.e., slowly vanishing $h_n$ as the sample size $n$ increases) should improve the convergence rate of the trimmed mean estimator.
It is important to ensure that our restriction on the structure of $(A,B)$ in Assumption \ref{a:m} does not rule out these cases where large trimming matters.
The following example demonstrates that Assumption \ref{a:m} is compatible with these relevant cases.

 \begin{example}[Heavy Tail]\label{ex:heavy_tail}
Let
$A \sim \text{Gamma}(\alpha,\beta)|_{[0,1]}$ (the truncated Gamma distribution to $[0,1]$) where $\alpha>0$ and $\beta>0$.
Let
$B \ | \ A=a \sim \mathcal{N}(\gamma a, \gamma_v a^\delta)$ where $\gamma \geq 0$ and $\gamma_v > 0$.
Note that the joint distribution of $(A,B)$ generated by the above process satisfies all four parts of Assumption \ref{a:m}.
Under this setting, simple calculations conclude that the moment $\mathbb{E}\left[B/A\right]$ exists if $2\alpha+\delta>2$ and the variance $Var(B/A)$ does \textit{not} exist if $\alpha + \delta < 2$.
Furthermore, we can deduce $n \left/ Var\left(\frac{B}{A} \cdot \mathbbm{1}\{A \geq h_n\}\right)\right. = O\left(nh_n^{2-\alpha-\delta}\right)$.
In other words, the convergence rate of the trimmed mean estimator $\widetilde\theta(h_n)$ is bounded below by $\sqrt{nh_n^{2-\alpha-\delta}}$.
Note that $\sqrt{nh_n^{2-\alpha-\delta}}$ can be arbitrarily close to $\sqrt{nh_n}$ when for instance $\delta=0$ and $\alpha\approx 1$.  
Therefore, larger trimming (i.e., more slowly vanishing $h_n$ as the sample size $n$ increases) helps to improve the convergence rate.
$\triangle$
\end{example}

\subsection{On Assumption \ref{a:h}}\label{sec:a:h}

The three parts (i), (ii), and (iii) of Assumption \ref{a:h} can be simultaneously satisfied if the order $k$ of smoothness assumed in Assumption \ref{a:m} (ii) is greater than 3.
In this case, the trimming threshold $h_n$ can be chosen to satisfy $n^{-\frac{1}{4}} \ll h_n \lesssim n^{-\frac{1}{2(k-1)}}$.
Therefore, the ``largest possible trimming'' or the slowest possible trimming threshold is achieved with a choice according to $h_n \sim n^{-\frac{1}{2(k-1)}}$.

 \subsection{On Assumption \ref{a:psi}}\label{sec:a:psi}

Assumption \ref{a:psi} imposes restrictions on how a researcher estimates the $k-1$ derivatives of the conditional expectation function $m$ of the numerator $B$ given the denominator $A$ at zero.
There are a number of alternative methods available in the literature, but we focus on the sieve estimation due to the availability of a rich set of theories that are relevant to our setting \citep{chen2007large}.\footnote{We take advantage of the asymptotic linear representation and the asymptotic distribution for least squares estimation of nonparametric regressions under i.i.d. settings \citep[e.g.,][]{van1990estimating,andrews1991asymptotic,eastwood1991adaptive,gallant1991asymptotic,newey1997convergence,de2002note,van2002m,huang2003local,chen2007large,cattaneo2013optimal,belloni2015some,chen2015optimal,hansen:2015}.}
In particular, Section \ref{sec:overview} proposes the shifted orthonormal Legendre polynomial basis of degree $K \geq k$ on $[0,1]$ -- see (\ref{eq:legendre}).
Our discussion of Assumption \ref{a:psi} here also focuses on the shifted orthonormal Legendre polynomial basis, although it is possible to use other sieve bases as well with minor modifications.

For the space of functions where $m$ resides, we consider the H\"older classes of smoothness order $s$, in addition to the requirements stated in Assumption \ref{a:m}.
The $k-1$ derivatives of $m$ at zero are estimated by the sieve predictor
\begin{align}\label{eq:mkappa}
\widehat m^{(\kappa)}(0) &= p_K^{(\kappa)}(0)' \mathbb{E}_n\left[ p_K(A)p_K(A)' \right]^{-1} \mathbb{E}_n\left[ p_K(A)B \right]
\end{align}
for each $\kappa \in \{1,...,k-1\}$.
The influence function for $\widehat m^{(\kappa)}(0)$ can be written as
\begin{align}\label{eq:psikappa}
\psi_\kappa = p_K^{(\kappa)}(0)' p_K(A) \left(B - m(A)\right)
\end{align}
for each $\kappa \in \{1,...,k-1\}$.

The estimators (\ref{eq:mkappa}) and their influence functions (\ref{eq:psikappa}) satisfy Assumption \ref{a:psi} provided that the following five conditions are satisfied:
\begin{enumerate}[(I)]
	\item Eigenvalues of $\mathbb{E}\left[p_K(A)p_K(A)'\right]$ are bounded above and away from zero.
	\item $n^{-1/2} h_n^{\kappa-1} \sqrt{\log K} \left(K + K^{5/2-s}\right)\norm{p_K^{(\kappa)}(0)} = o(1)$ for each $\kappa \in \{1,...,k-1\}$;
	\item $h_n^{\kappa-1} K^{1-s} \norm{p_K^{(\kappa)}(0)} = o(1)$ for each $\kappa \in \{1,...,k-1\}$;
	\item $n^{1/2} h_n^{\kappa-1} \abs{r_K^{(\kappa)}(0)} = o(1)$ for each $\kappa \in \{1,...,k-1\}$; and
	\item $n^{-1/4} h_n^{\kappa-1} \mathbb{E}\left[\left(p_K^{(\kappa)}(A)' p_K(A) \left(B - m(A)\right)\right)^4\right]^{1/4} = O(1)$;
\end{enumerate}
where $r_K^{(\kappa)}(0)$ is the sieve approximation error given by
$$
r_K^{(\kappa)}(0) = m^{(\kappa)}(0) - p_K^{(\kappa)}(0)' \mathbb{E}\left[ p_K(A)p_K(A)' \right]^{-1} \mathbb{E}_n\left[ p_K(A)m(A) \right].
$$
Conditions (I)--(IV) suffice for Assumption \ref{a:psi} (i) according to \citet{belloni2015some} for the case of the shifted orthonormal Legendre polynomial basis $p_K$.
Assumption \ref{a:psi} (ii) is satisfied by the concrete influence function expression in (\ref{eq:psikappa}).
Assumption \ref{a:psi} (iii) is satisfied by the concrete influence function expression in (\ref{eq:psikappa}) and condition (V).

\section{Properties of the Estimator}\label{sec:properties}
\subsection{Negligible Variance of the Bias Estimator}\label{sec:negligible_variance_of_the_bias_estimator}

The purpose of trimming in practice is to reduce the variance of the estimator, but this variance reduction is at the cost of the trimming bias.
We thus propose in this paper to estimate this trimming bias, yet it would put our priorities backwards if the bias estimator would incur a larger variance than that of the original untrimmed estimator.
In this section, we present sufficient conditions under which the variance of the bias estimator converges to zero faster than the variance of the trimmed mean estimator, which is analogous to the result of \citet[Theorem 3.4]{chaudhuri/hill:2016} whose bias estimator does not contribute to the asymptotic variance.
In other words, the bias estimator adds only an asymptotically negligible variance under suitable conditions.

For simplicity, we suppose throughout this subsection that $1/A$ follows the Pareto distribution with the shape parameter $\alpha > 0$.
With the scale normalization, the cumulative distribution function $F_A$ of $A$ takes the form
$$
F_A(a)=a^{\alpha}
$$
for $a \in [0,1]$.
Note that we allow for $\alpha>0$ as opposed to just $\alpha>1$.
The following proposition shows sufficient conditions under which, as argued in the previous paragraph, the bias estimator will add only a negligible variance in large samples, and hence the bias correction will improve the root mean square error in large samples.
A major implication of this finding is that the combination of trimming and its bias correction entails asymptotically no larger root mean squared error than a na\"ive untrimmed mean estimator under the suggested conditions.
This result holds even if $\alpha$ is less than the unit.\footnote{When $\alpha$ is less than the unit, $B/A$ is not integrable so $\mathbb{E}[B/A]$ does not exist. Our parameter of interest, introduced in (\ref{eq:definition_parameter_of_interest}), is  well-defined even in this case.}

 \begin{proposition}[Negligible Variance of the Bias Estimator]\label{proposition:variance_ratio}
If 
the conditional second moment function $v$, defined by
$
v(a)=\mathbb{E}[B^2\mid A=a],
$
is three times differentiable in a neighborhood of zero with $v(0) \neq 0$ or $v^{(1)}(0) \neq 0$, and there exists $\alpha>0$ such that 
$
F_A(a)=a^{\alpha}
$
and
$
\sum_{\kappa=1}^{k-1}\mathbb{E}\left[\left(h_n^{\kappa-1}\psi_\kappa\right)^2\right]=O(h_n^{-1}),
$
then 
$$
\frac{Var(c(h_n)' \psi)}{Var(\frac{B}{A} \cdot \mathbbm{1}\{A \geq h_n\})}=o(1).
$$
\end{proposition}
 
A proof of this proposition is given in Appendix \ref{sec:proposition:variance_ratio}.

\subsection{On Faster Convergence Rates}\label{sec:faster}

Our motivation in this paper is that smoothness in the joint distribution between $A$ and $B$ can be exploited to achieve faster convergence rates.
This section formalizes this claim.

When we assume the smoothness conditions stated in Assumption \ref{a:m} for sufficiently large $k$, our proposed estimator $\widehat\theta(h_n)$ can have a convergence rate arbitrarily close to $n^{-1/2}$. 
In this section, we show that, for every $\eta>0$, the convergence rate of $\widehat\theta(h_n)$ is faster than $n^{-1/2+\eta}$ for sufficiently large $k$. 

\begin{proposition}\label{theorem:faster}
Let $\eta$ be any positive number.
Suppose $h_n=Cn^{-\eta}$ for some constant $C>0$.
If Assumptions \ref{a:m} and \ref{a:psi} are satisfied with $k\geq 1+\frac{1}{2\eta}$ and 
\begin{equation}\label{eq:sieve_estimate_fast}
\mathbb{E}\left[\|\psi\|^2\right]=O(n^{2\eta}),
\end{equation}
then  
$$
\widehat\theta(h_n)-\theta_0=O_p(n^{-1/2+\eta}).
$$
\end{proposition} 

A proof of this proposition is given in Appendix \ref{sec:proof:faster}.
This proposition assumes that the estimation error of 
$\left(m^{(1)}(0),...,m^{(k-1)}(0)\right)$ is close to the parametric rate such that its influence function $\psi$ satisfies \eqref{eq:sieve_estimate_fast}. 
A sufficient condition for bounding the estimation error of 
$\left(m^{(1)}(0),...,m^{(k-1)}(0)\right)$ depends on its specific estimator. 
In the sieve estimation, for example, the point-wise convergence rate can be arbitrarily close to $n^{-1/2}$ as long as we restrict the function class for $m(\cdot)$.  Detailed discussions are found in \citet{belloni2015some}.

\subsection{Data-Dependent Trimming Threshold}\label{sec:data_dependent} 

The main asymptotic distribution theory (Theorem \ref{theorem:asymptotic_distribution}) is based on a rate condition (Assumption \ref{a:h}) imposed on  the trimming threshold.
In practice, researchers may choose the trimming threshold through data-driven choice rules, as emphasized by \cite{escanciano2014uniform,escanciano2016identification} in the context of different econometric models.
In this section, we present a range of possibly random trimming thresholds under which the limit distribution of Theorem \ref{theorem:asymptotic_distribution} continues to be true.

 \begin{proposition}[Asymptotic Distribution with Data-Dependent Trimming Threshold]\label{theorem:data_dependent}
Suppose that the assumptions in Theorem \ref{theorem:asymptotic_distribution} hold. 
If there are deterministic sequences, $\{h_{n,L}\}_{n\geq 1}$ and $\{h_{n,U}\}_{n\geq 1}$, such that 
\begin{eqnarray}\label{eq:h_{n,U}_small}
&\mathbb{E}\left[\mathbbm{1}\{A<h_{n,U}\}\right]=o(1),
\\
\label{eq:h_{n,U}_h_{n,L}_veryysmall}
&
\mathbb{E}\left[B^2\cdot \mathbbm{1}\{h_{n,L}\leq A<h_{n,U}\}\right]
=
o(h_{n,L}^2),
\\
\label{eq:bandwidth_estimation1}
&h_{n,L}\leq h_n\leq h_{n,U},
\qquad\text{and} 
\\
\label{eq:bandwidth_estimation2}
&h_{n,L}\leq\widehat{h}_n\leq h_{n,U}\mbox{ with probability approaching one},
\end{eqnarray}
then
\begin{align*}
\frac{ \widehat\theta(\widehat{h}_n) - \theta_0 }{\sqrt{\sigma(h_n)^2/n}}
\stackrel{d}{\rightarrow} \mathcal{N}(0,1),
\end{align*}
provided that $\sigma(h_n)^2\equiv Var(Z(h_n))$ is bounded away from zero.
\end{proposition}

A proof of this proposition is given in Appendix \ref{sec:proof:data_dependent}.
The condition \eqref{eq:h_{n,U}_h_{n,L}_veryysmall} deserves a remark. This condition requires that $h_{n,L}$ and $h_{n,U}$ are not too far way from $h_n$ such that we can establish the uniform convergence in the proof. 
This condition can be simplified further; for example, if $\mathbb{E}[B^2\mid A=a]$ is bounded over $a\in[0,1]$ and $F_A(a)=a^\alpha$ for some $\alpha>0$, then \eqref{eq:h_{n,U}_h_{n,L}_veryysmall} holds with $(h_{n,L},h_{n,U})=(h_n-h_n^3,h_n+h_n^3)$.\footnote{This statement follows from 
$
\mathbb{E}\left[B^2\cdot \mathbbm{1}\{h_{n,L}\leq A<h_{n,U}\}\right]
\leq
\left(\sup_{a}\mathbb{E}[B^2\mid A=a]\right)((h_n+h_n^3)^{\alpha}-(h_n-h_n^3)^{\alpha})
=
O(h_n^{2+\alpha})
=
O(h_{n,L}^{2+\alpha}).
$}

\section{Implementation}\label{sec:implementation}

\subsection{Variance Estimation}\label{sec:variance_estimation}

The asymptotic distribution in Theorem \ref{theorem:asymptotic_distribution} depends on the unknown parameter $Var(Z(h_n))$.
To conduct a feasible inference, we estimate this theoretical variance by 
\begin{align*}
\widehat\sigma(h_n)^2
=
\mathbb{E}_n\left[\widehat{Z}(h_n)^2\right]-\mathbb{E}_n\left[\widehat{Z}(h_n)\right]^2
\end{align*}
where 
\begin{align}\label{eq:Z(h_n)at}
\widehat{Z}(h_n) = \frac{B}{A} \cdot \mathbbm{1}\{A \geq h_n\} + \widehat{c}(h_n)' \widehat{\psi},
\end{align}
and we assume that the influence function $\psi$ of $\left(\widehat m^{(1)}(0),...,\widehat m^{(k-1)}(0)\right)$ and its estimator $\widehat\psi$ satisfy the following conditions.
\newtheorem*{a_psi_prime}{Assumption \ref{a:psi}$'$}
\begin{a_psi_prime}
(i)  $\mathbb{E}_n\left[\norm{\psi}^2\right]^{1/2}=O_p\left(n^{1/2}\right)$; and
(ii) $\mathbb{E}\left[\norm{\widehat{\psi}-\psi}^2\right]^{1/2}=o(1)$.
\end{a_psi_prime}

Part (i) of this assumption allows the empirical $L^2$ norm of the influence function to grow up to $n^{1/2}$.
Part (ii) of this assumption requires that the influence function estimator is consistent in mean square.
We now state a corollary to Theorem \ref{theorem:asymptotic_distribution}, arguing that the asymptotic distribution is the same even after replacing the theoretical variance $Var(Z(h_n))$ by its estimate $\widehat\sigma(h_n)^2$.

 \begin{corollary}[Asymptotic Distribution with Estimated Variance]\label{corollary:asymptotic_distribution}
If Assumptions \ref{a:m}, \ref{a:h}, \ref{a:psi}, and \ref{a:psi}$'$ are satisfied, then
\begin{align*}
\frac{ \widehat\theta(h_n) - \theta_0 }{\sqrt{\left.\widehat\sigma(h_n)^2\right/n}}
\stackrel{d}{\rightarrow} \mathcal{N}(0,1),
\end{align*}
provided that $Var(Z(h_n))$ is bounded away from zero.
\end{corollary}
 
See Appendix \ref{sec:corollary:asymptotic_distribution} for a proof.

\subsection{Choice of $h_n$}\label{sec:choice_of_h}

The literature on denominator-based-trimmed mean estimation suggests rule-of-thumb choices of $h_n$ \citep[e.g.,][]{crump2009dealing}.
In this section, we also discuss a choice rule based on optimal approximate mean square errors.
Despite the optimal choice, the higher-order bias correction based on our bias characterization result can in turn provide a valid inference in the spirit of \citet{calonico2014robust}.
By balancing the leading bias and the variance of the trimmed-mean estimator with bias correction up to the order of $k-2$, we can choose the trimming threshold $h_n$ to minimize 
\begin{align*}
\left(\frac{\mathbb{E}\left[A^{k-2} \cdot \mathbbm{1}\{0<A<h_n\}\right]}{(k-1)!} \cdot m^{(k-1)}(0)\right)^2
+
n^{-1} \cdot Var\left(\frac{B}{A} \cdot \mathbbm{1}\left\{A \geq h_n\right\} + \widetilde c(h_n)' \widetilde \psi\right),
\end{align*}
where $\widetilde c(h_n)$ is the $(k-2)$-dimensional vector defined by
\begin{align*}
\widetilde c(h_n) = 
\left(\begin{array}{c}
\mathbb{E}\left[\mathbbm{1}\{0<A<h_n\}\right]/1!\\
\mathbb{E}\left[A \cdot \mathbbm{1}\{0<A<h_n\}\right]/2!\\
\vdots\\
\mathbb{E}\left[A^{k-3} \cdot \mathbbm{1}\{0<A<h_n\}\right]/(k-2)!
\end{array}\right)
\end{align*}
and $\widetilde\psi$ is the influence function of an estimator $\left(\widehat m^{(1)}(0),...,\widehat m^{(k-2)}(0)\right)$ of the $k-2$ derivatives $\left( m^{(1)}(0),..., m^{(k-2)}(0)\right)$ of the conditional expectation function $m$ at 0.
Since our asymptotic theory is based on the bias-corrected estimator with bias estimation up to the order of $k-1$, the asymptotic variance is valid in large samples under the above choice of $h_n$. 

\subsection{Choice of $K$}\label{sec:choice_of_K}
For a choice of the degree $K$, a researcher can follow common rules of thumb used in sieve estimation.
Alternatively, a data-driven choice of $K$ may be obtained by common model selection approaches, such as cross validation and generalized information criterion.
We present a model selection approach based on the Bayesian Information Criterion (BIC) for guidance.
For each $K > 3$, define the BIC criterion
$$
BIC(K) = n \ln\left(\frac{1}{n}\sum_{i=1}^n (\widehat B_i(K) - B_i )^2\right) + K \ln(n),
$$
where
$$
\widehat B_i(K) = p_K(A_i)' \mathbb{E}_n[p_K(A)p_K(A)']^{-1} \mathbb{E}_n[p_K(A)B].
$$
The BIC-optimal choice of $K$ is given by
$$
K = \arg \min_{\widetilde K > 3} BIC(\widetilde K).
$$

\section{Extension to Cases of Generated Random Variables}\label{sec:extension}

In some applications, the random variables $A$ and $B$ are generated.
For example, the propensity score $A=P$ in Example \ref{ex:ipsw} is often estimated by the logit or probit in empirical research.
In this light, we now consider the case where $(A,B) = (g_A(X,\gamma_0),g_B(X,\gamma_0))$ is a function of observed random vector $X$ and an unknown finite-dimensional parameter vector $\gamma_0 \in \Gamma$.
We write the denominator-based-trimmed mean estimator as
\begin{align}\label{eq:ex:theta_h_hat}
\widetilde\theta(h_n)
= \mathbb{E}_n\left[\frac{g_B(X,\widehat\gamma)}{g_A(X,\widehat\gamma)} \cdot S\left(\frac{g_A(X,\widehat\gamma)}{h_n}\right)\right],
\end{align}
where 
$\widehat\gamma$ is an estimator of $\gamma_0$ and
$S$ is a smoothed indicator function such that $S(0)=0$ and $S(u)=1$ for $u \geq 1$. 
Unlike the baseline case where we used the non-smooth indicator for trimming, we now employ smooth trimming in order to apply the delta method to approximate the effects of first-stage estimation by $\widehat\gamma$ on the bias-corrected denominator-based-trimmed mean estimator.
For convenience, we also write the population counterpart of the denominator-based-trimmed mean estimator (\ref{eq:ex:theta_h_hat}) by
\begin{align}\label{eq:ex:theta_h}
\theta(h_n)
= \mathbb{E}\left[\frac{B}{A} \cdot S\left(\frac{A}{h_n}\right)\right]
= \mathbb{E}\left[\frac{g_B(X,\gamma_0)}{g_A(X,\gamma_0)} \cdot S\left(\frac{g_A(X,\gamma_0)}{h_n}\right)\right].
\end{align}
Consider the following assumptions.

 \begin{assumption}\label{a:ex:m}
(i) $m(0) = 0$;
(ii) $m$ is $k$-times continuously differentiable with a bounded $k$-th derivative in a neighborhood of 0;
(iii) $g_A$ and $g_B$ are twice differentiable in $\gamma$ in a neighborhood of $\gamma_0$;
(iv) $\underset{\gamma\in\Gamma}{\sup}\norm{\frac{\partial}{\partial\gamma}g_A(X,\gamma)}^2$ and $\underset{\gamma\in\Gamma}{\sup}\norm{\frac{\partial^2}{\partial\gamma\partial\gamma'}g_A(X,\gamma)}$ have finite expectations;
(v)  $\norm{\left.\frac{\partial}{\partial \gamma} g_B(X,\gamma)\right\vert_{\gamma=\gamma_0}}^2$,
$g_B(X,\gamma_0)^2 \cdot  \norm{\left.\frac{\partial}{\partial \gamma} g_A(X,\gamma)\right\vert_{\gamma=\gamma_0}}^2$,
$\sup_{\gamma \in \Gamma}\norm{\frac{\partial}{\partial\gamma'} g_A(X,\gamma)}\norm{\frac{\partial}{\partial\gamma} g_B(X,\gamma)}$, 
$\sup_{\gamma \in \Gamma}\norm{\frac{\partial^2}{\partial\gamma\partial\gamma'} g_B(X,\gamma)}$, 
$\sup_{\gamma \in \Gamma}|g_B(X,\gamma)|\cdot \norm{\frac{\partial^2}{\partial\gamma\partial\gamma'} g_A(X,\gamma)}$, 
and 
$\sup_{\gamma \in \Gamma}|g_B(X,\gamma)|\cdot \norm{\frac{\partial}{\partial\gamma} g_A(X,\gamma)}^2$
 have   finite expectations; and 
(vi) $P(g_A(X,\gamma_0)=0)=0$.
\end{assumption}

\begin{assumption}\label{a:ex:s}
(i) $S(0)=0$, $S(u)=1$ for $u \geq 1$, $\sup_u \abs{S(u)} < \infty$;
(ii) $S$ is twice differentiable;
(iii) $\sup_u \abs{\frac{d^2}{d u^2} \frac{S(u)}{u}} < \infty$;
(iv) $\sup_u \abs{\frac{d^2}{d u^2} S(u)} < \infty$.
\end{assumption}

\begin{assumption}\label{a:ex:h}
(i) $nh_n^{2(k-1)} = O(1)$; and 
(ii) $n^{-1}h_n^{-6} = o(1)$
as $n \rightarrow \infty$.
\end{assumption}
 
Assumption \ref{a:ex:m} is a counterpart of Assumption \ref{a:m} from the baseline framework, and Assumption \ref{a:ex:h} is a counterpart of Assumption \ref{a:h} from the baseline framework.
Assumption \ref{a:ex:s} did not appear in the baseline framework.
As mentioned earlier, this assumption imposes the smoothness of $S$ so we can apply the delta method to approximate the effects of first-stage estimation by $\widehat\gamma$ on the bias-corrected denominator-based-trimmed mean estimator.
Similar to Theorem \ref{lemma:bias} for the baseline setting, we have the following theorem stating a bias characterization under the current extended setting.

 \begin{theorem}[Bias Characterization]\label{lemma:ex:bias}
If Assumptions \ref{a:ex:m} (i)--(ii), \ref{a:ex:s} (i), and \ref{a:ex:h} (i) are satisfied, then
\begin{align*}
\theta_0 - \theta(h_n)
=
-\lambda(h_n)
+
o\left(n^{-1/2}\right),
\end{align*}
where
\begin{align}\label{eq:ex:bias}
\lambda(h_n)
=
-\sum_{\kappa=1}^{k-1} \frac{\mathbb{E}\left[A^{\kappa-1} \left(S\left(\frac{A}{h_n}\right)-1\right)\right]}{\kappa !} \cdot m^{(\kappa)}(0).
\end{align}
\end{theorem}
 
The proof of this theorem is provided in Appendix \ref{sec:lemma:ex:bias}.
The bias expression (\ref{eq:ex:bias}) in Theorem \ref{lemma:ex:bias} motivates the bias estimator
\begin{align*}
\widehat \lambda(h_n)
=
-\sum_{\kappa=1}^{k-1} \frac{\mathbb{E}_n\left[g_A(X,\widehat\gamma)^{\kappa-1} \left(S\left(\frac{g_A(X,\widehat\gamma)}{h_n}\right)-1\right)\right]}{\kappa !} \cdot \widehat m^{(\kappa)}(0;\widehat\gamma),
\end{align*}
where 
\begin{align*}
\widehat m^{(\kappa)}(0;\gamma) = p_K^{(\kappa)}(0)' \mathbb{E}_n\left[ p_K(g_A(X,\widehat\gamma))p_K(g_A(X,\widehat\gamma))' \right]^{-1} \mathbb{E}_n\left[ p_K(g_A(X,\widehat\gamma))g_B(X,\widehat\gamma) \right]
\end{align*}
for each $\kappa \in \{1,...,k-1\}$.
Thus, the bias-corrected denominator-based-trimmed mean estimator reads
\begin{align*}
\widehat\theta(h_n)
=& \mathbb{E}_n\left[\frac{g_B(X,\widehat\gamma)}{g_A(X,\widehat\gamma)} \cdot S\left(\frac{g_A(X,\widehat\gamma)}{h_n}\right)\right]
\\
&\quad+
\sum_{\kappa=1}^{k-1} \frac{\mathbb{E}_n\left[g_A(X,\widehat\gamma)^{\kappa-1} \left(S\left(\frac{g_A(X,\widehat\gamma)}{h_n}\right)-1\right)\right]}{\kappa !} \cdot \widehat m^{(\kappa)}(0;\widehat\gamma).
\end{align*}
We next consider the following assumptions.

\begin{assumption}\label{a:ex:psi}
(i)
$\left(1,...,h_n^{k-2}\right)' \circ
\left(
\widehat m^{(1)}(0;\gamma_0) - m^{(1)}(0), ...,
\widehat m^{(k-1)}(0;\gamma_0) - m^{(k-1)}(0)
\right)'
=
\left(1,...,h_n^{k-2}\right)' \circ \mathbb{E}_n\left[\psi\right] + o_p\left(n^{-1/2}\right)$;
(ii)
$\mathbb{E}\left[\psi\right]=0$;
(iii) $\widehat m^{(\kappa)}(0; \ \cdot \ )$ is twice differentiable for each $\kappa \in \{1,...,k-1\}$;
(iv) $\left.\frac{\partial}{\partial\gamma'} \widehat m^{(\kappa)}(0;\gamma)\right|_{\gamma=\gamma_0} - \left.\frac{\partial}{\partial\gamma'} m^{(\kappa)}(0;\gamma)\right|_{\gamma=\gamma_0} = o_p(1)$ for each $\kappa \in \{1,...,k-1\}$;
(vi) $\underset{\gamma\in\Gamma}{\sup}\norm{\frac{\partial^2}{\partial\gamma\partial\gamma'} \widehat m^{(\kappa)}(0;\gamma)} $ $= o_p\left(n^{1/2}\right)$ for each $\kappa \in \{1,...,k-1\}$.
\end{assumption}

\begin{assumption}\label{a:ex:gamma}
$\widehat\gamma - \gamma_0 = \mathbb{E}_n[\phi] + O_p(n^{-1})$
where
$\mathbb{E}[\phi]=0$ and $\mathbb{E}[\phi^2]<\infty$.
\end{assumption}

Assumption \ref{a:ex:psi} is a counterpart of Assumption \ref{a:psi} from the baseline framework.
Assumption \ref{a:ex:gamma} is new to the current framework, and imposes a mild restriction on the first-stage estimator $\widehat\gamma$ of $\gamma_0$.
Under the above assumptions, we show in Lemma \ref{lemma:ex:influence_function} that the influence function of $\widehat\theta(h_n)-\theta(h_n)-\lambda(h_n)$ can be written as
\begin{align*}
Z(h_n) 
=&
 \cdot \omega_{1,n}(X,\gamma_0)
+\mathbb{E}\left[\left.\frac{\partial}{\partial\gamma}\omega_{1,n}(X,\gamma)\right|_{\gamma=\gamma_0}\right]' \cdot \phi
\\
&+
\sum_{\kappa=1}^{k-1} \frac{\omega_{2,\kappa,n}(X,\gamma_0) \cdot m^{(\kappa)}(0) + \mathbb{E}[\omega_{2,\kappa,n}(X,\gamma_0)] \cdot \psi_\kappa}{\kappa!}
\\
&+
\sum_{\kappa=1}^{k-1} \frac{m^{(\kappa)}(0) \cdot \mathbb{E}\left[\left.\frac{\partial}{\partial\gamma'}\omega_{2,\kappa,n}(X,\gamma)\right|_{\gamma=\gamma_0}\right] + \mathbb{E}[\omega_{2,\kappa,n}(X,\gamma_0)] \cdot \left.\frac{\partial}{\partial\gamma'} m^{(\kappa)}(0;\gamma)\right|_{\gamma=\gamma_0}}{\kappa!} \cdot \phi
\end{align*}
where
\begin{align}
\omega_{1,n}(x,\gamma) 
&= g_B(X,\gamma)\frac{S\left(\frac{g_A(X,\gamma)}{h_n}\right)}{g_A(X,\gamma)}
\qquad\text{and}
\label{eq:ex:omega1}
\\
\omega_{2,\kappa,n}(x,\gamma)
&= g_A(X,\gamma)^{\kappa-1} \cdot \left(S\left(\frac{g_A(X,\gamma)}{h_n}\right) - 1\right).
\label{eq:ex:omega2}
\end{align}
Therefore, using the bias characterization in Theorem \ref{lemma:ex:bias}, we obtain the following asymptotic distribution result similarly to Theorem \ref{theorem:asymptotic_distribution} from the baseline case.

\begin{theorem}[Asymptotic Distribution]\label{theorem:ex:asymptotic_distribution}
If Assumptions \ref{a:ex:m}, \ref{a:ex:s}, \ref{a:ex:h}, \ref{a:ex:psi}, and \ref{a:ex:gamma} are satisfied, then 
$$
\frac{\widehat\theta(h_n)-\theta_0}{\sqrt{Var(Z(h_n))/n}}\stackrel{d}{\rightarrow} \mathcal{N}(0,1)
$$
holds,
provided that $Var(Z(h_n))$ is bounded away from zero and 
\begin{align}\label{eq:ex:l4l2}
\frac{\mathbb{E}\left[\left(Z(h_n)-\mathbb{E}\left[Z(h_n)\right]\right)^{2+\delta}\right]}{n^{\delta/2} \mathbb{E}\left[\left(Z(h_n)-\mathbb{E}\left[Z(h_n)\right]\right)^2\right]^{(2+\delta)/2}} = o\left(1\right)
\end{align}
for some $\delta>0$.
\end{theorem}
 
The proof of this theorem is provided in Appendix \ref{sec:theorem:ex:asymptotic_distribution}.
The condition (\ref{eq:ex:l4l2}) in the statement of the above theorem was proved as Lemma \ref{lemma:l4l2} in the baseline case.
In order to accommodate a general class of first-stage estimators $\widehat\gamma$ and its interaction with $(A,B)$, however, we state this condition as a high-level assumption in the current extended framework.
Primitive sufficient conditions for this condition will vary depending on applications of interest, but the condition can be proved by similar lines of argument to the proof of Lemma \ref{lemma:l4l2} in Appendix \ref{sec:lemma:l4l2}.

\section{Simulation Studies}\label{sec:simulation}

In this section, we present simulation analysis for the proposed method.
Consider the following extension to the setup described in Example \ref{ex:heavy_tail}.
Let
$A \sim \text{Gamma}(0.50+\epsilon,0.50)|_{[0,1]}$ (the truncated Gamma distribution to $[0,1]$), where $\epsilon$ is varied across simulations.
Let
$B \ | \ A=a \sim \mathcal{N}\left(\sum_{j=1}^5 \gamma_j a^j, 0.50^2\cdot a\right)$ where  $\left(\gamma_1,...,\gamma_5\right)'=\left(0.50,...,0.10\right)'$.
Throughout, the sample size is set to $n=500$ and the trimming threshold $h_n$ is fixed at $0.01$, which is about the 10th percentile of the distribution of $\text{Gamma}(0.50+0.01,,0.50)|_{[0,1]}$.
We compare the performance of the untrimmed sample mean $\mathbb{E}_n\left[B/A\right]$, the trimmed mean $\widetilde\theta(h_n)$, and the bias-corrected trimmed mean $\widehat\theta(h_n)$ in this setting.

From the discussions in Example \ref{ex:heavy_tail}, it is for the values of 
$\epsilon \in (0.00,0.50)$ that the moment $\mathbb{E}\left[B/A\right]$ is well-defined but the variance $Var\left(B/A\right)$ does not exist.
In other words, it is for this range where trimming improves the convergence rate of an estimator for $\mathbb{E}\left[B/A\right]$.
Therefore, we vary $\epsilon$ in this range for our simulation studies.
Specifically, we run simulations for values of $\epsilon \in \{0.10,0.20,0.30,0.40\}$.
The case of $\epsilon=0.10$ yields the heaviest tail for the distribution of $B/A$, where $\mathbb{E}\left[B/A\right]$ is well-defined while the variance does not exist.
As $\epsilon$ increases, the tail of the distribution of $B/A$ becomes less heavy.

The sieve dimension $K$ and the smoothness order $k$ are set to $K=k \in \{4,5\}$.
Recall from the discussion of Assumption \ref{a:h} in Section \ref{sec:a:h} that an admissible choice of the trimming threshold is possible for $k>3$.
We thus consider $K=k \geq 4$.
Note that the true data generating process has the fifth-order polynomial for $m$.
The case of $K=4$, therefore, entails a non-zero approximation error by the fourth-order sieve estimator.
On the other hand, the fifth-order sieve in the case of $K=5$ correctly specifies the fifth-order polynomial $m$.
The purpose of using both these cases, $K=4$ and $K=5$, is to examine the potential effects of sieve approximation error on the bias-corrected trimmed mean estimator $\widehat\theta(h_n)$.

We first check the conclusion of Theorem \ref{theorem:asymptotic_distribution} by simulations with these data generating designs.
Recall that Theorem \ref{theorem:asymptotic_distribution} states that the statistic $\left(\widehat\theta(h_n)-\theta_0\right)/\sqrt{\sigma(h_n)^2/n}$ converges in distribution to a standard normal random variable.
Figure \ref{fig:kernel_density} illustrates kernel density plots of $\left(\widehat\theta(h_n)-\theta_0\right)/\sqrt{\sigma(h_n)^2/n}$  based on 100,000 Monte Carlo iterations under alternative data generating parameters as solid curves.
We also draw the density of the limit standard normal distribution in dashed curves for convenience of making comparisons.
Notice that the kernel densities are reasonably close to the standard normal density in each case.
These results support the statement of Theorem \ref{theorem:asymptotic_distribution}.
In particular, these results demonstrate that the sampling variations in the higher-order derivative estimator $\left(\widehat m^{(1)}(0),...,\widehat m^{(k-1)}(0)\right)$ of the conditional expectation function $m(\cdot)$ are incorporated in the asymptotic property stated in Theorem \ref{theorem:asymptotic_distribution}.

\begin{figure}
	\centering
	\begin{tabular}{cc}
		$K=k=4$ and $\epsilon=0.10$ & $K=k=5$ and $\epsilon=0.10$ \\
		\includegraphics[width=0.45\textwidth]{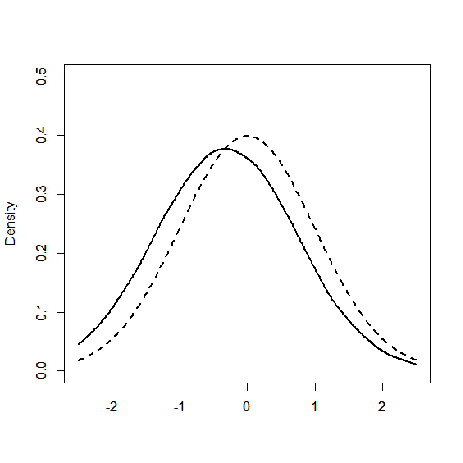} &
		\includegraphics[width=0.45\textwidth]{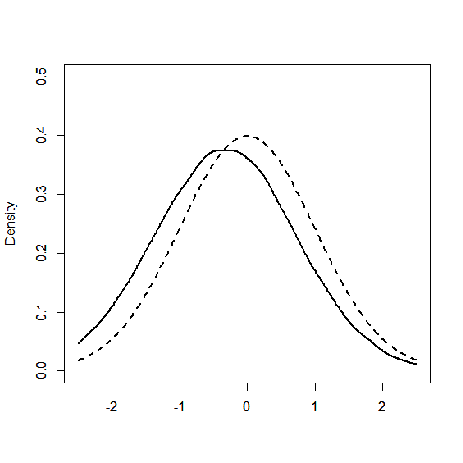} \\
		$K=k=4$ and $\epsilon=0.40$ & $K=k=5$ and $\epsilon=0.40$ \\				
		\includegraphics[width=0.45\textwidth]{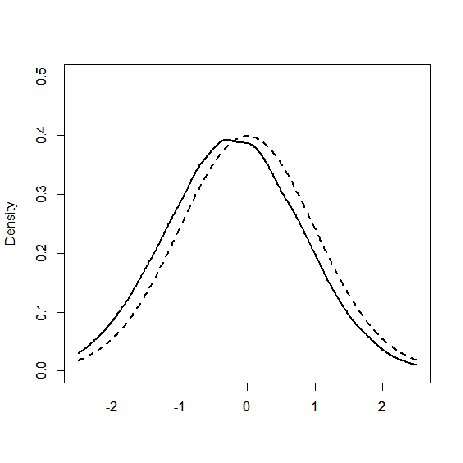} &
		\includegraphics[width=0.45\textwidth]{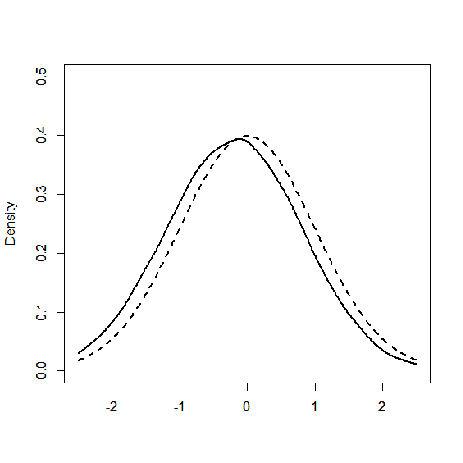}
	\end{tabular}
	\caption{Solid curves depict kernel density plots of the simulated statistics $\left(\widehat\theta(h_n)-\theta_0\right)/\sqrt{\sigma(h_n)^2/n}$. Dashed curves depict the density of the limit standard normal distribution. To top row displays results under $\epsilon = 0.1$ while the bottom row displays results under $\epsilon = 0.4$. The left column displays results under $K=k=4$, while the right column displays results under $K=k=5$. The bandwidth choice follows Silverman's rule.}
	\label{fig:kernel_density}
\end{figure}

We next present estimation and inference results.
Table \ref{tab:simulation} summarizes simulation results based on 100,000 Monte Carlo iterations.
The simulated average, the simulated bias, the simulated root mean squared error (RMSE), and the simulated coverage frequency for the nominal level of 5 percent (95\% Coverage) are reported for each of (I) the untrimmed sample mean $\mathbb{E}_n\left[B/A\right]$, (II) the trimmed mean $\widetilde\theta(h_n)$, (III) the bias-corrected trimmed mean $\widehat\theta(h_n)$ under $K=k=4$, and (IV) the bias-corrected trimmed mean $\widehat\theta(h_n)$ under $K=k=5$.

\begin{table}
	\centering
		\begin{tabular}{lrlccccc}
\hline\hline
$\epsilon$
&&&& 0.100 & 0.200 & 0.300 & 0.400\\
\hline
True $\theta_0$
&&&& 0.409 & 0.400 & 0.392 & 0.385\\
\hline
Average
&  (I)&Sample Mean    && 0.191 & 0.415 & 0.386 & 0.378\\
& (II)&Trimmed Mean   && 0.368 & 0.372 & 0.372 & 0.370\\
&(III)&BCTM ($K=k=4$) &&\cellcolor{gray!15} 0.388 &\cellcolor{gray!15} 0.384 &\cellcolor{gray!15} 0.379 &\cellcolor{gray!15} 0.375\\
& (IV)&BCTM ($K=k=5$) &&\cellcolor{gray!15} 0.388 &\cellcolor{gray!15} 0.384 &\cellcolor{gray!15} 0.379 &\cellcolor{gray!15} 0.375\\
\hline
Bias
&  (I)&Sample Mean    &&-0.218 & 0.015 &-0.006 &-0.006\\
& (II)&Trimmed Mean   &&-0.041 &-0.028 &-0.020 &-0.015\\
&(III)&BCTM ($K=k=4$) &&\cellcolor{gray!15}-0.021 &\cellcolor{gray!15}-0.016 &\cellcolor{gray!15}-0.013 &\cellcolor{gray!15}-0.010\\
& (IV)&BCTM ($K=k=5$) &&\cellcolor{gray!15}-0.021 &\cellcolor{gray!15}-0.016 &\cellcolor{gray!15}-0.012 &\cellcolor{gray!15}-0.010\\
\hline
RMSE
&  (I)&Sample Mean    &&57.638 & 7.248 & 0.721 & 0.161\\
& (II)&Trimmed Mean   && 0.076 & 0.066 & 0.060 & 0.054\\
&(III)&BCTM ($K=k=4$) &&\cellcolor{gray!15} 0.071 &\cellcolor{gray!15} 0.063 &\cellcolor{gray!15} 0.058 &\cellcolor{gray!15} 0.053\\
& (IV)&BCTM ($K=k=5$) &&\cellcolor{gray!15} 0.072 &\cellcolor{gray!15} 0.064 &\cellcolor{gray!15} 0.058 &\cellcolor{gray!15} 0.053\\
\hline
95\% Coverage
&  (I)&Sample Mean    && 0.973 & 0.967 & 0.961 & 0.956\\
& (II)&Trimmed Mean   && 0.898 & 0.924 & 0.933 & 0.939\\
&(III)&BCTM ($K=k=4$) &&\cellcolor{gray!15} 0.927 &\cellcolor{gray!15} 0.938 &\cellcolor{gray!15} 0.941 &\cellcolor{gray!15} 0.944\\
& (IV)&BCTM ($K=k=5$) &&\cellcolor{gray!15} 0.922 &\cellcolor{gray!15} 0.936 &\cellcolor{gray!15} 0.942 &\cellcolor{gray!15} 0.945\\
\hline\hline
		\end{tabular}
	\caption{Simulation results based on 100,000 Monte Carlo iterations. The displayed statistics include the simulated average, the simulated bias, the simulated root mean squared error (RMSE), and the simulated coverage frequency for the nominal level of 5 percent (95\% Coverage). Estimators used in simulations include (I) the untrimmed sample mean $\mathbb{E}_n\left[B/A\right]$, (II) the trimmed mean $\widetilde\theta(h_n)$, (III) the bias-corrected trimmed mean $\widehat\theta(h_n)$ under $K=k=4$, and (IV) the bias-corrected trimmed mean $\widehat\theta(h_n)$ under $K=k=5$.}
	\label{tab:simulation}
\end{table}

Observe that the untrimmed sample mean (I) produces huge values of the root mean squared error, particularly when $\epsilon$ is small, i.e., when the distribution of $B/A$ has the heaviest tail.\footnote{The reported `bias' of the untrimmed sample mean (I) is also large for $\epsilon=0.100$, but this is because the simulated bias (which differs from the `bias' formally defined by $\mathbb{E}[\breve{\theta}]-\mathbb{E}[B/A]$ for an estimator $\breve{\theta}$) suffers from large outliers in the heavy-tailed distribution.}
On the other hand, the trimmed mean (II) and the bias corrected trimmed means, (III) and (IV), achieve much smaller values of the root mean squared error.\footnote{On average, 33 observations (6.6\%) are trimmed for the case of the heaviest tail, i.e., $\epsilon=0.10$.}
Among these three well-performing estimators in terms of the root mean squared error, the bias-corrected estimators, (III) and (IV), further achieve smaller biases, smaller mean squared error, and more accurate coverages than the uncorrected estimator (II).
There seems little difference between the estimator (III) based on $K=k=4$ and the estimator (IV) based on $K=k=5$.
In summary, from these observations, we conclude that the bias-corrected trimmed means perform better than the untrimmed sample mean or the trimmed mean.
The effect of approximation error by the parsimonious specification $K=k=4$ in (III) is negligible compared with the fully specified estimator (IV).

\section{Empirical Illustration}\label{sec:empirical}

We illustrate our robust inference method with an empirical analysis of the dependence of firms on external finance in the U.S. manufacturing industry for the years 2000--2010. 
The external dependence for a firm is defined as a ratio of external borrowing to capital expenditures.
It indicates the extent to which manufacturing firms need external finance. 
This ratio has been used in the economic growth literature -- for example, \cite{RaZi1998} use the external dependence to identify which industrial sectors can benefit more from financial development. 

We use the Compustat database (Fundamentals Annual) for the years 2000--2010. 
We focus on U.S. manufacturing firms excluding the food industry, i.e., Standard Industry Classification (SIC) Code ranging between 2100--3999. 
Following \cite{RaZi1998}, we construct a proxy for the firm's dependence via the ratio of external borrowing to capital expenditures. 
Specifically, the variable $A$ represents the capital expenditures (CAPX), and the variable $B$ represents the capital expenditures (CAPX) minus cash flows from operations (FOPT).\footnote{As in \cite{RaZi1998}, for the firms with the format code 7 (SCF=7), we cannot observe FOPT and therefore impute the sum of IBC, DPC, TXDC, ESUBC, SPPIV, and FOPO instead.} 
The observations with missing $A$ or $B$ are dropped.
The samples consist of sizes $n=$ 2008, 1844, 1771, 1699, 1645, 1583, 1481, 1386, 1289, 1214 and 1179 for the years 2000--2010, respectively.

Figure \ref{fig:fig_mean_i} displays estimates of the mean and median of $B/A$ on the vertical axis plotted against the list of years 2000--2010 on the horizontal axis.
The shades indicate 99\% confidence intervals.
The top left graph shows the result for the mean $\mathbb{E}[B/A]$ using the na\"ive sample mean $\mathbb{E}_n[B/A]$.
The top right graph shows the result for the median using the sample median.
The bottom left graph shows the result for the mean $\mathbb{E}[B/A]$ using the trimmed sample mean $\widetilde\theta(h_n)$ without bias correction.
The bottom right graph shows the result for the mean $\mathbb{E}[B/A]$ using our proposed method of the bias-corrected trimmed sample mean $\widehat\theta(h_n)$.
Here, $h_n$ and $K$ are chosen according to the procedure in Sections \ref{sec:choice_of_h}--\ref{sec:choice_of_K}.
Since the result for $\mathbb{E}_n[B/A]$ has a range different from the others in Figure \ref{fig:fig_mean_i}, we provide a rescaled version of the figure for $\mathbb{E}_n[B/A]$ in Figure \ref{fig:fig_mean_wide}.

\begin{figure}
	\centering
	\begin{tabular}{cc}
		\includegraphics[width=0.450\textwidth]{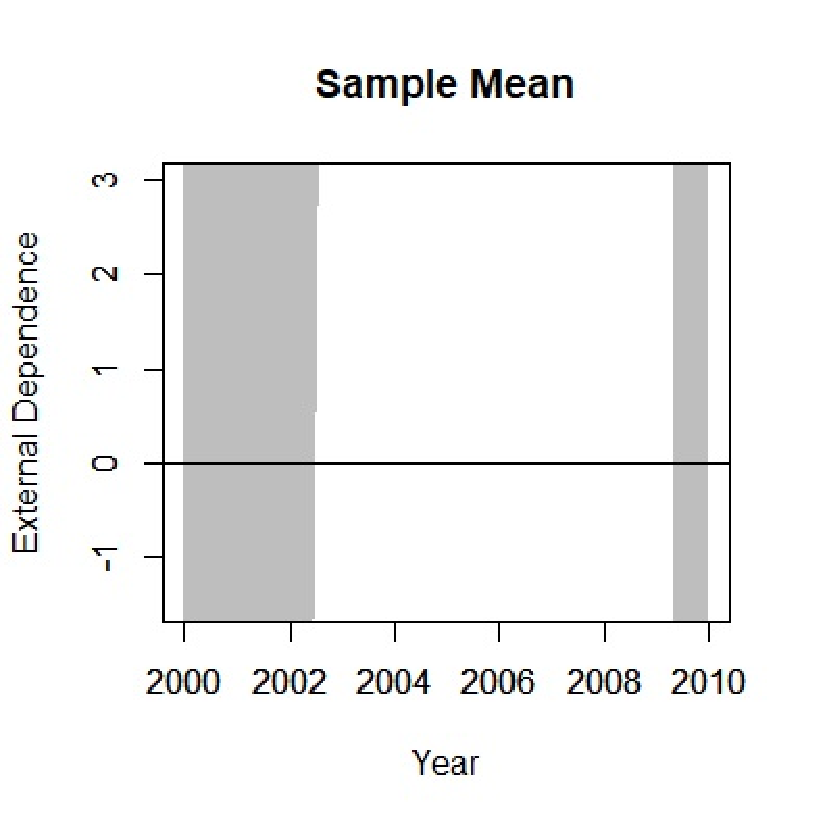} &
		\includegraphics[width=0.450\textwidth]{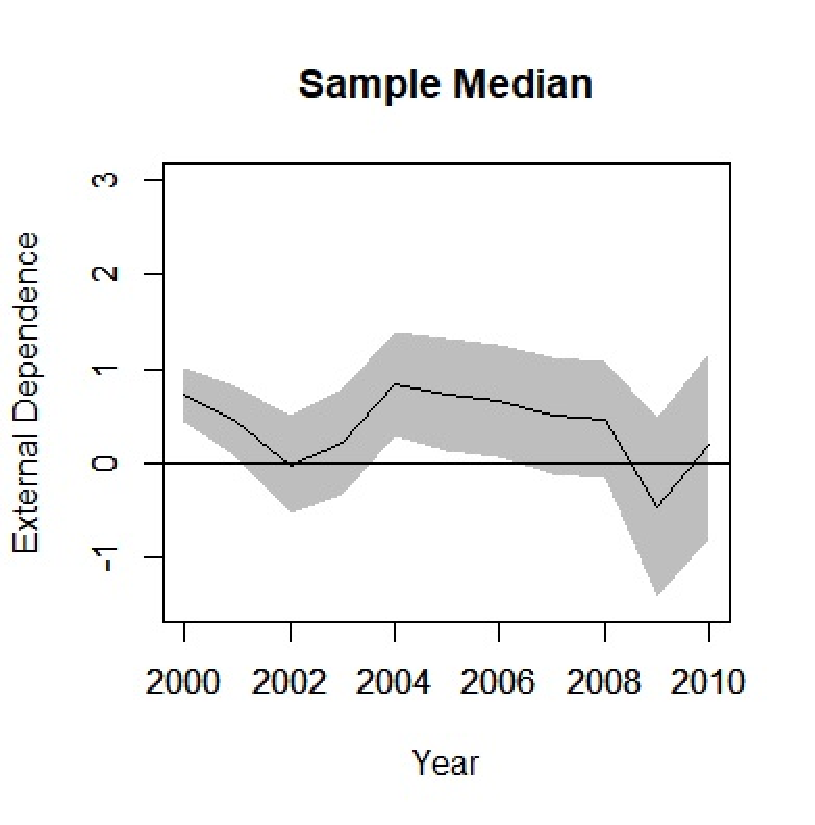} \\
		\includegraphics[width=0.450\textwidth]{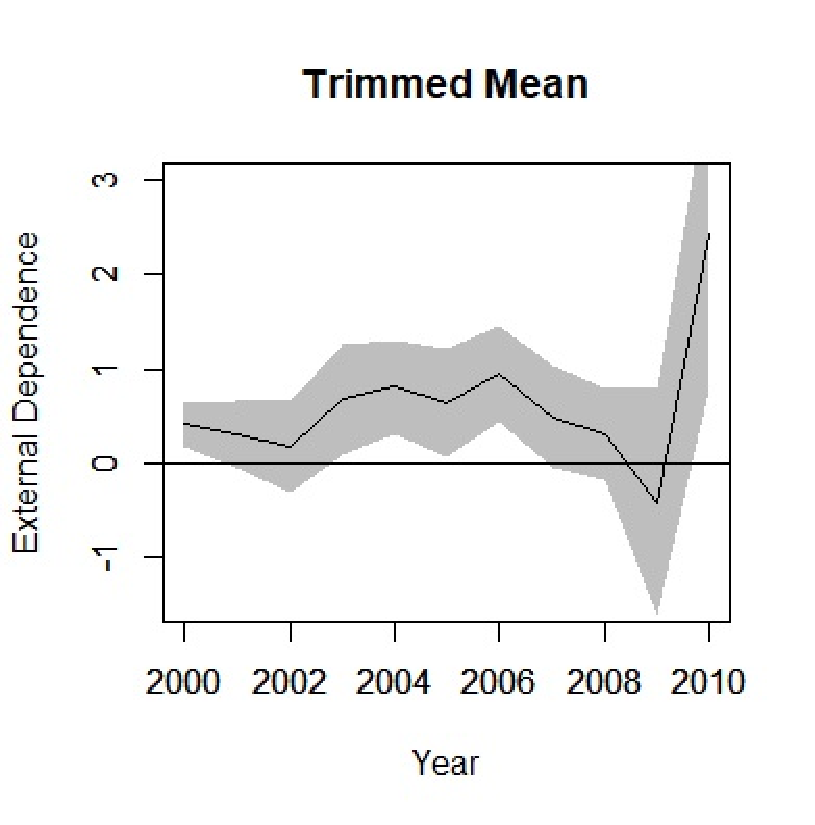} &
		\includegraphics[width=0.450\textwidth]{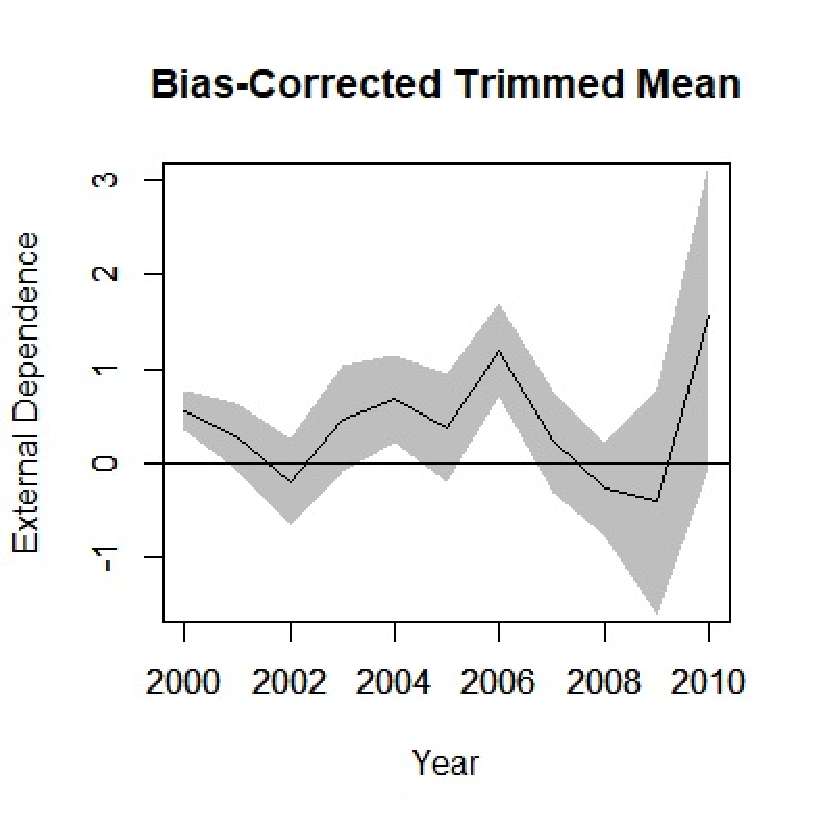}
	\end{tabular}
	\caption{99\% confidence intervals for the mean/median of the external financial dependence in the U.S. manufacturing industry for years 2000--2010. The top left graph shows the result for the mean $\mathbb{E}[B/A]$ using the na\"ive sample mean $\mathbb{E}_n[B/A]$, the top right graph shows the result for the median using the sample median, the bottom left graph shows the result for the mean $\mathbb{E}[B/A]$ using the trimmed sample mean $\widetilde\theta(h_n)$, and the bottom right graph shows the result for the mean $\mathbb{E}[B/A]$ using the bias-corrected trimmed mean $\widehat\theta(h_n)$. $h_n$ and $K$ are chosen according to the procedure in Sections \ref{sec:choice_of_h}--\ref{sec:choice_of_K}.}
	\label{fig:fig_mean_i}
\end{figure}

\begin{figure}
\centering
\includegraphics[width=0.450\textwidth]{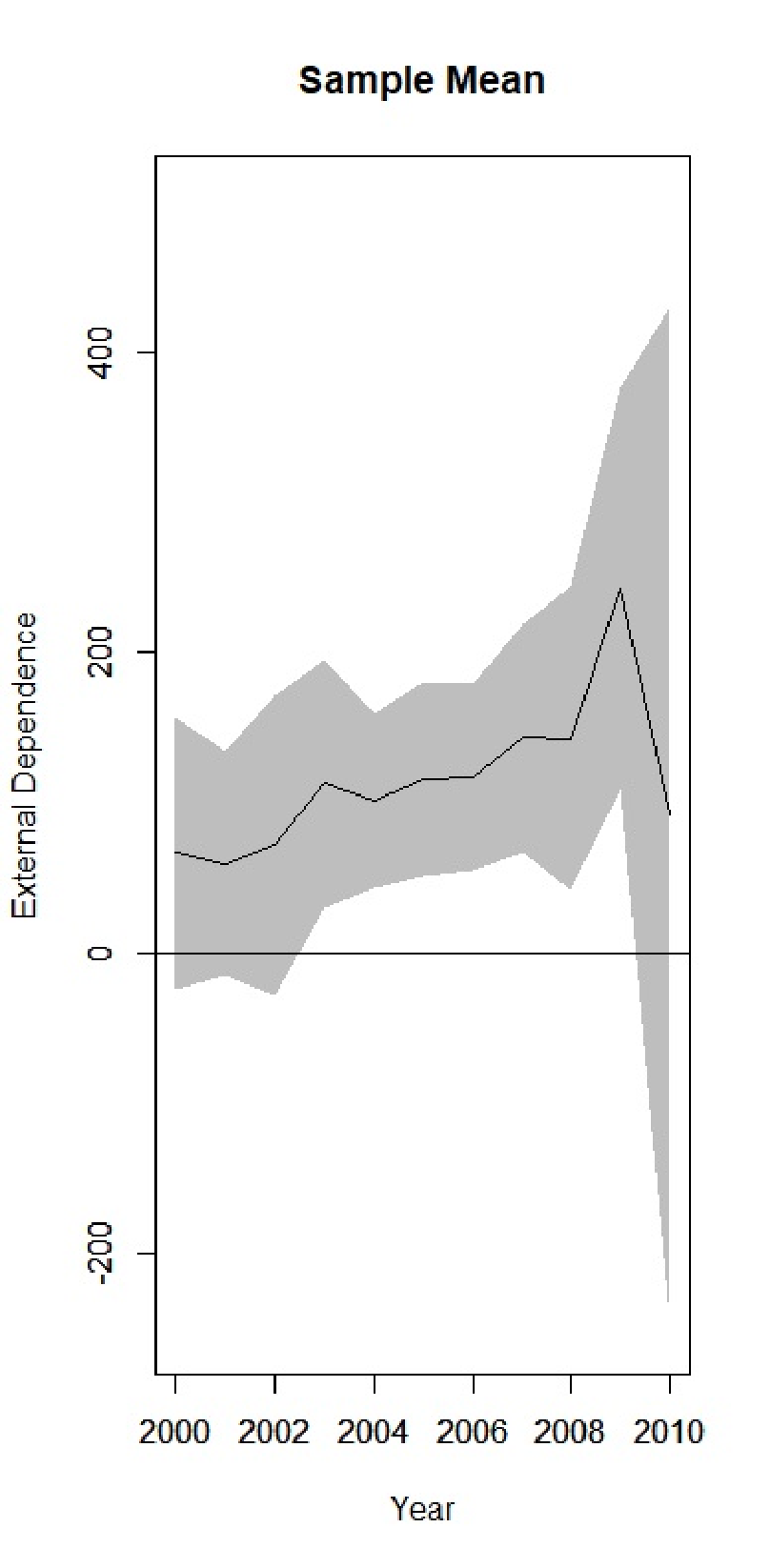} 	
\caption{The top left graph in Figure \ref{fig:fig_mean_i} with a widened range of external dependence. Displayed are the na\"ive sample mean $\mathbb{E}_n[B/A]$ and its 99\% confidence interval. 
}
\label{fig:fig_mean_wide}
\end{figure}

Not surprisingly, the na\"ive sample mean yields huge confidence intervals, and does not provide informative inference results.
The median on the other hand is robust against outliers, and it provides more informative inference results.
The trimmed sample mean without bias correction shows that the mean is in fact rather close to the median than that indicated by the na\"ive sample mean.
The bias-corrected trimmed sample mean exhibits a slightly different trajectory than the trimmed sample mean without bias correction, and shows that external financial dependence is smaller during the times of recessions, such as 2002 and the Great Recession of 2007--2009.

\section{Conclusion}\label{sec:conclusion}

This paper proposes a new method of estimation and inference for moments of ratios of the form $\mathbb{E}[B/A]$. 
By taking advantage of the smooth joint distribution between the numerator $B$ and the denominator $A$, we propose to correct the bias of a denominator-based-trimmed mean of $B/A$ with a wider range of trimming thresholds.
The wider range offers greater insensitivity of valid inference to the exact trimming level.
Precisely, the admissible upper bound of the trimming threshold $h_n$ can be as slow as to the order of $n^{-1/(2(k-1))}$.
Further, the validity of inference under such large trimming thresholds entails that reduced variances can be achieved without incurring non-negligible trimming bias. 
We formally show, as its consequence, that greater smoothness of the conditional expectation function $a \mapsto \mathbb{E}[B|A=a]$ allows for faster convergence rates of the estimator.
We demonstrate that the proposed method is applicable to the cases where the distribution of $B/A$ has a heavy tail and thus a na\"ive sample mean converges slowly.

\newpage
\appendix
\section*{Appendix}
\section{Proofs of the Main Results}\label{sec:proofs_of_the_main_results}

\subsection{Proof of Theorem \ref{lemma:bias}: Bias Characterization}\label{sec:lemma:bias}
\begin{proof}
For $a$ small enough, the Taylor expansion of $m$ around 0 yields
\begin{align}\label{eq:taylor}
m(a) 
&= 
\sum_{\kappa=1}^{k-1} \frac{a^\kappa}{\kappa!} \cdot m^{(\kappa)}(0) + \frac{a^k}{(k-1)!} \cdot \int_0^1 (1-t)^{k-1} m^{(k)}(ta)dt
\end{align}
under Assumption \ref{a:m} (i)--(ii).
We can now rewrite $\theta_0-\theta(h_n) = \mathbb{E}\left[\frac{B}{A} \cdot \mathbbm{1}\{0<A<h_n\}\right]$ as
\begin{align*}
\mathbb{E}\left[\frac{B}{A} \cdot \mathbbm{1}\{ 0<A<h_n \}\right]
&=
\mathbb{E}\left[\frac{m(A)}{A} \cdot \mathbbm{1}\{ 0<A<h_n \}\right]
\\
&=
\sum_{\kappa=1}^{k-1} \frac{\mathbb{E}\left[A^{\kappa-1} \cdot \mathbbm{1}\{ 0<A<h_n \}\right]}{\kappa!} \cdot m^{(\kappa)}(0)\\ &\quad+ \frac{\mathbb{E}\left[A^{k-1} \cdot \int_0^1 (1-t)^{k-1} m^{(k)} (tA) dt \cdot \mathbbm{1}\{ 0<A<h_n \}\right]}{(k-1)!}
\end{align*}
where the first equality follows from the law of iterated expectations, and
the second equality follows from (\ref{eq:taylor}).
Therefore, the claimed equality follows.
\end{proof}

\subsection{Proof of Theorem \ref{theorem:asymptotic_distribution}: Asymptotic Distribution}\label{sec:theorem:asymptotic_distribution}
\begin{proof}
First, note that the equality derived in Lemma \ref{lemma:l4l2} in Appendix \ref{sec:lemma:l4l2} is equivalent to Lyapunov's condition
\begin{align*}
\frac{\mathbb{E}\left[\left(Z(h_n)-\mathbb{E}\left[Z(h_n)\right]\right)^{2+\delta}\right]}{n^{\delta/2} \mathbb{E}\left[\left(Z(h_n)-\mathbb{E}\left[Z(h_n)\right]\right)^2\right]^{(2+\delta)/2}} = o\left(1\right)
\end{align*}
in particular for $\delta = 2 > 0$.
Therefore, by Lyapunov's central limit theorem, we have
\begin{align}
\frac{ (\mathbb{E}_n-\mathbb{E}) \left[Z(h_n)\right] }{\sqrt{Var(Z(h_n))/n}}
\stackrel{d}{\rightarrow} \mathcal{N}(0,1)
\label{eq:theorem:asymptotic_distribution1}
\end{align}
under Assumptions \ref{a:m} (iii)--(iv), \ref{a:h}, and \ref{a:psi} (iii) invoked in Lemma \ref{lemma:l4l2}.
Next, observe that
\begin{align}
\frac{ (\mathbb{E}_n-\mathbb{E}) \left[Z(h_n)\right] }{\sqrt{Var(Z(h_n))/n}}
&=
\frac{ \widehat\theta(h_n) - \theta_0 }{\sqrt{Var(Z(h_n))/n}}
+
\frac{ o_p\left(1\right) }{\sqrt{Var(Z(h_n))}}
=
\frac{ \widehat\theta(h_n) - \theta_0 }{\sqrt{Var(Z(h_n))/n}}
+
o_p(1),
\label{eq:theorem:asymptotic_distribution2}
\end{align}
where the first equality is due to Lemma \ref{lemma:linear_representation} in Appendix \ref{sec:lemma:linear_representation} under Assumptions \ref{a:m} (i)--(iii), \ref{a:h} (i), and \ref{a:psi},
and the second equality holds because $Var(Z(h_n))$ is bounded away from zero.
Applying Slutsky's theorem to (\ref{eq:theorem:asymptotic_distribution1}) and (\ref{eq:theorem:asymptotic_distribution2}) yields the result claimed in the theorem.
\end{proof}

\subsection{Proof of Proposition \ref{proposition:variance_ratio}: Negligible Variance of the Bias Estimator}\label{sec:proposition:variance_ratio}

\begin{proof}
Since $Var(\frac{B}{A} \cdot \mathbbm{1}\{A \geq h_n\})$ is bounded away from zero and $\mathbb{E}[\frac{B}{A} \cdot \mathbbm{1}\{A \geq h_n\}]=O(1)$, it suffices to show that 
\begin{align}
\label{eq:lemma:variance_ratio1}
\frac{\mathbb{E}[(c(h_n)' \psi)^2]}{\mathbb{E}[\frac{B^2}{A^2} \cdot \mathbbm{1}\{A \geq h_n\}]}=o(1).
\end{align}
First, note that
\begin{align*}
\mathbb{E}\left[A^{\kappa-1}\mathbbm{1}\{0<A<h_n\}\right]
=
\alpha \cdot
\int_0^{h_n}a^{\kappa-1}a^{\alpha-1}da
=
\alpha \cdot
\int_0^{h_n}a^{\alpha+\kappa-2}da
=
\frac{\alpha \cdot h_n^{\alpha+\kappa-1}}{\alpha+\kappa-1},
\end{align*}
and therefore
\begin{align}
\mathbb{E}[(c(h_n)' \psi)^2]^{1/2}
&\leq
\sum_{\kappa=1}^{k-1}\mathbb{E}\left[\left(\frac{\mathbb{E}\left[A^{\kappa-1}\mathbbm{1}\{0<A<h_n\}\right]}{\kappa!}\psi_\kappa\right)^2\right]^{1/2}
\notag\\
&=
\sum_{\kappa=1}^{k-1}\mathbb{E}\left[\left(\frac{\alpha \cdot h_n^{\alpha+\kappa-1}}{\kappa!\cdot (\alpha+\kappa-1)}\psi_\kappa\right)^2\right]^{1/2}
\notag\\
&=
\sum_{\kappa=1}^{k-1}\frac{\alpha \cdot h_n^{\alpha}}{\kappa!\cdot (\alpha+\kappa-1)}\mathbb{E}\left[\left(h_n^{\kappa-1}\psi_\kappa\right)^2\right]^{1/2},
\label{eq:lemma:variance_ratio2}
\end{align}
where the inequality is due to Minkowski's inequality.
The rest of the proof branches into two cases depending on $\alpha \geq 1$ or $\alpha < 1$.
\bigskip\\
{\bf Case 1:} $\alpha\geq 1$.
In this case,
$$
\mathbb{E}\left[\left(c(h_n)' \psi\right)^2\right]
=
\sum_{\kappa=1}^{k-1}O\left(h_n^{2\alpha}\right) \cdot \mathbb{E}\left[\left(h_n^{\kappa-1}\psi_\kappa\right)^2\right]=\sum_{\kappa=1}^{k-1}O\left(h_n^{1+\alpha}\right) \cdot \mathbb{E}\left[\left(h_n^{\kappa-1}\psi_\kappa\right)^2\right]=o(1),
$$ 
where
the first equality is due to (\ref{eq:lemma:variance_ratio2}),
the second equality is due to $\alpha \geq 1$, and
the last equality follows from the condition in the statement of the lemma.
This establishes (\ref{eq:lemma:variance_ratio1}).
\bigskip\\
{\bf Case 2:} $\alpha<1$.
The denominator of (\ref{eq:lemma:variance_ratio1}) can be computed as follows.
\begin{align*}
\mathbb{E}\left[\frac{B^2}{A^2} \cdot \mathbbm{1}\{A \geq h_n\}\right]
=&
\mathbb{E}\left[\frac{v(A)}{A^2} \cdot \mathbbm{1}\{A \geq h_n\}\right]\\
=&
\alpha \cdot 
\int_{h_n}^1v(a)a^{\alpha-3}da\\
=&
\alpha \cdot 
v(0) \cdot \int_{h_n}^1a^{\alpha-3}da
+\alpha \cdot 
v^{(1)}(0) \cdot \int_{h_n}^1a^{\alpha-2}da\\&
+
\frac{\alpha \cdot v^{(2)}(0)}{2} \cdot \int_{h_n}^1a^{\alpha-1}da
+
\frac{\alpha}{2} \cdot \int_{h_n}^1 \int_0^1 (1-t)^2 \cdot v^{(3)}(ta) \cdot a^{\alpha}dtda\\
=&
\frac{\alpha\cdot v(0) \cdot (1-h_n^{\alpha-2})}{\alpha-2}
+\frac{\alpha\cdot v^{(1)}(0) \cdot (1-h_n^{\alpha-1})}{\alpha-1}\\&
+\frac{v^{(2)}(0)\cdot (1-h_n^\alpha)}{2}
+\frac{\alpha}{2} \cdot \int_{h_n}^1 \int_0^1 (1-t)^2 \cdot v^{(3)}(ta) \cdot a^{\alpha}dtda\\
=&
-\frac{\alpha\cdot v(0) \cdot h_n^{\alpha-2}}{\alpha-2}
-\frac{\alpha\cdot v^{(1)}(0) \cdot h_n^{\alpha-1}}{\alpha-1}
+\frac{\alpha\cdot v(0)}{\alpha-2}\\&
+\frac{\alpha\cdot v^{(1)}(0)}{\alpha-1} 
+\frac{v^{(2)}(0)}{2}
+\frac{\alpha}{2} \cdot \int_0^1 \int_0^1 (1-t)^2 \cdot v^{(3)}(ta) \cdot a^{\alpha}dtda\\&
+o(1),
\end{align*}
where
the first equality is due to the law of iterated expectations,
the second equality is due to $F_A(a)=a^{\alpha}$,
the third equality is due to Taylor expansion under the condition in the statement of the lemma,
and the remaining lines execute calculations.
Since $v(0) \neq 0$ or $v^{(1)}(0) \neq 0$ by the condition in the statement of the lemma, it follows that
$$
\mathbb{E}\left[\frac{B^2}{A^2} \cdot \mathbbm{1}\{A \geq h_n\}\right]^{-1}=O\left(h_n^{1-\alpha}\right).
$$
By (\ref{eq:lemma:variance_ratio2}), therefore,
$$
\frac{\mathbb{E}[(c(h_n)' \psi)^2]}{\mathbb{E}[\frac{B^2}{A^2} \cdot \mathbbm{1}\{A \geq h_n\}]}=
\sum_{\kappa=1}^{k-1}O(h_n^{1+\alpha})\mathbb{E}\left[\left(h_n^{\kappa-1}\psi_\kappa\right)^2\right]
=
o(1).
$$ 
This establishes (\ref{eq:lemma:variance_ratio1}).
\end{proof}

\subsection{Proof of Proposition \ref{theorem:faster}: Faster Convergence Rates}\label{sec:proof:faster}

\begin{proof}
Note that Assumption \ref{a:h} holds for the given sequence $h_n$ and given $k$. 
By Lemma \ref{lemma:linear_representation}, it suffices to show 
$\mathbb{E}\left[Z(h_n)^2\right]^{1/2}=O(n^{\eta})$.
By the triangular inequality, 
\begin{eqnarray*}
\mathbb{E}\left[Z(h_n)^2\right]^{1/2} 
&\leq&
\mathbb{E}\left[ \frac{B^2}{A^2} \cdot \mathbbm{1}\{A \geq h_n\}\right]^{1/2} +\mathbb{E}\left[( c(h_n)' \psi)^2\right]^{1/2}.
\end{eqnarray*}
Since $\mathbb{E}\left[B^2\cdot \mathbbm{1}\{A \geq h_n\}\right]=O(1)$, we have 
$$
\mathbb{E}\left[ \frac{B^2}{A^2} \cdot \mathbbm{1}\{A \geq h_n\}\right]^{1/2} 
=
O(h_n^{-1})
=
O(n^{\eta}).
$$
Moreover, 
$$
\mathbb{E}\left[( c(h_n)' \psi)^2\right]^{1/2}
\leq
\|c(h_n)\|\mathbb{E}\left[\|\psi\|^2\right]^{1/2}
=
O(n^{\eta}),
$$
where the last equality uses $\|c(h_n)\|=O(1)$.
\end{proof}

\subsection{Proof of Proposition \ref{theorem:data_dependent}: Data-Dependent Trimming Threshold}\label{sec:proof:data_dependent}

\begin{proof}
By Theorem \ref{theorem:asymptotic_distribution} and $1/\sigma(h_n)=O(1)$, it suffices to show $\widehat\theta(\widehat{h}_n) - \widehat\theta(h_n)=o_p(n^{-1/2})$.
By Eq. \eqref{eq:bandwidth_estimation1} and \eqref{eq:bandwidth_estimation2}, it is sufficient to show  
$$
\sup_{\check{h}^1,\check{h}^2\in[h_{n,L},h_{n,U}]}|\widehat\theta(\check{h}^1) - \widehat\theta(\check{h}^2)|=o_p(n^{-1/2}).
$$
By Eq. \eqref{eq:estimator}
\begin{eqnarray*}
\widehat\theta(\check{h}^1)-\widehat\theta(\check{h}^2)
&=&
\mathbb{E}_n\left[\left(\frac{B}{A}- \sum_{\kappa=1}^{k-1} \frac{A^{\kappa-1}}{\kappa!} \cdot \widehat m^{(\kappa)}(0)\right) \cdot \mathbbm{1}\{\check{h}^1\leq A<\check{h}^2\}\right]
\\&&\quad
-\mathbb{E}_n\left[\left(\frac{B}{A}- \sum_{\kappa=1}^{k-1} \frac{A^{\kappa-1}}{\kappa!} \cdot \widehat m^{(\kappa)}(0)\right) \cdot \mathbbm{1}\{\check{h}^2\leq A<\check{h}^1\}\right]
\\
&=&
\mathbb{E}_n\left[\left(\frac{B}{A}- \sum_{\kappa=1}^{k-1} \frac{A^{\kappa-1}}{\kappa!} \cdot  m^{(\kappa)}(0)\right) \cdot \mathbbm{1}\{\check{h}^1\leq A<\check{h}^2\}\right]
\\&&\quad
-\mathbb{E}_n\left[\left(\frac{B}{A}- \sum_{\kappa=1}^{k-1} \frac{A^{\kappa-1}}{\kappa!} \cdot  m^{(\kappa)}(0)\right) \cdot \mathbbm{1}\{\check{h}^2\leq A<\check{h}^1\}\right]
\\&&\quad
-
\sum_{\kappa=1}^{k-1} \mathbb{E}_n\left[\frac{A^{\kappa-1}}{\kappa!}\cdot \mathbbm{1}\{\check{h}^1\leq A<\check{h}^2\}\right] \cdot (\widehat m^{(\kappa)}(0)- m^{(\kappa)}(0)) 
\\&&\quad
+\sum_{\kappa=1}^{k-1}\mathbb{E}_n\left[  \frac{A^{\kappa-1}}{\kappa!} \cdot \mathbbm{1}\{\check{h}^2\leq A<\check{h}^1\}\right] \cdot (\widehat m^{(\kappa)}(0)-m^{(\kappa)}(0)).
\end{eqnarray*}
By Assumption \ref{a:psi}, it follows that
\begin{eqnarray*}
\widehat\theta(\check{h}^1)-\widehat\theta(\check{h}^2)
&=&
\mathbb{E}_n\left[\left(\frac{B}{A}- \sum_{\kappa=1}^{k-1} \frac{A^{\kappa-1}}{\kappa!} \cdot  m^{(\kappa)}(0)\right) \cdot \mathbbm{1}\{\check{h}^1\leq A<\check{h}^2\}\right]
\\&&\quad
-\mathbb{E}_n\left[\left(\frac{B}{A}- \sum_{\kappa=1}^{k-1} \frac{A^{\kappa-1}}{\kappa!} \cdot  m^{(\kappa)}(0)\right) \cdot \mathbbm{1}\{\check{h}^2\leq A<\check{h}^1\}\right]
\\&&\quad
+o_p(n^{-1/2}).
\end{eqnarray*}

First, we are going to show 
$$
\sup_{\check{h}^1,\check{h}^2\in[h_{n,L},h_{n,U}]}\left|\mathbb{E}\left[\left(\frac{B}{A}- \sum_{\kappa=1}^{k-1} \frac{A^{\kappa-1}}{\kappa!} \cdot  m^{(\kappa)}(0)\right)\cdot  \mathbbm{1}\{\check{h}^1\leq A<\check{h}^2\}\right]\right| 
=
o(n^{-1/2}).
$$
By the law of iterated expectations and Eq. \eqref{eq:taylor}, 
\begin{eqnarray*}
&&
\left|\mathbb{E}\left[\left(\frac{B}{A}- \sum_{\kappa=1}^{k-1} \frac{A^{\kappa-1}}{\kappa!} \cdot  m^{(\kappa)}(0)\right)\cdot  \mathbbm{1}\{\check{h}^1\leq A<\check{h}^2\}\right]\right| 
\\
&=&
\left|\mathbb{E}\left[\left(\frac{m(A)}{A}- \sum_{\kappa=1}^{k-1} \frac{A^{\kappa-1}}{\kappa!} \cdot  m^{(\kappa)}(0)\right)\cdot  \mathbbm{1}\{\check{h}^1\leq A<\check{h}^2\}\right]\right| \\ 
&=& 
\left|\mathbb{E}\left[\frac{A^{k-1}}{(k-1)!} \cdot \int_0^1 (1-t)^{k-1} m^{(k)}(tA)dt \cdot  \mathbbm{1}\{\check{h}^1\leq A<\check{h}^2\}\right]\right|. 
\end{eqnarray*}
Therefore, we have 
\begin{eqnarray*}
&&
\left|\mathbb{E}\left[\left(\frac{B}{A}- \sum_{\kappa=1}^{k-1} \frac{A^{\kappa-1}}{\kappa!} \cdot  m^{(\kappa)}(0)\right)\cdot  \mathbbm{1}\{\check{h}^1\leq A<\check{h}^2\}\right]\right| 
\\
&&\quad\leq
\frac{h_n^{k-1} }{(k-1)!} \cdot \sup_{a\leq h}\left|\int_0^1 (1-t)^{k-1} m^{(k)}(ta)dt \right|\cdot \mathbb{E}\left[\mathbbm{1}\{A<\check{h}^2\}\right],
\end{eqnarray*}
and thus 
\begin{eqnarray*}
&&
\sup_{\check{h}^1,\check{h}^2\in[h_{n,L},h_{n,U}]}\left|\mathbb{E}\left[\left(\frac{B}{A}- \sum_{\kappa=1}^{k-1} \frac{A^{\kappa-1}}{\kappa!} \cdot  m^{(\kappa)}(0)\right)\cdot  \mathbbm{1}\{\check{h}^1\leq A<\check{h}^2\}\right]\right|\\
 &&\quad=
O(h_n^{k-1} \mathbb{E}\left[\mathbbm{1}\{A<h_{n,U}\}\right])\\ 
&&\quad=
o(n^{-1/2}),
\end{eqnarray*}
where the first equality is due to Assumption \ref{a:m} (i)--(ii), and the second equality is due to Assumptions \ref{a:m} (iii) and \ref{a:h} (i).

Next, we are going to show 
$$
\sup_{\check{h}^1,\check{h}^2\in[h_{n,L},h_{n,U}]}\left|(\mathbb{E}_n-\mathbb{E})\left[\left(\frac{B}{A}- \sum_{\kappa=1}^{k-1} \frac{A^{\kappa-1}}{\kappa!} \cdot  m^{(\kappa)}(0)\right) \cdot \mathbbm{1}\{\check{h}^1\leq A<\check{h}^2\}\right]\right|=o_p(n^{-1/2}).
$$
Define the function class 
$$
\mathcal{F}=\left\{(A,B)\mapsto\left(\frac{B}{A}- \sum_{\kappa=1}^{k-1} \frac{A^{\kappa-1}}{\kappa!} \cdot  m^{(\kappa)}(0)\right) \cdot \mathbbm{1}\{\check{h}^1\leq A<\check{h}^2\}: \left(\check{h}^1,\check{h}^2\right)\in[h_{n,L},h_{n,U}]^2\right\}
$$ 
with the envelope $F$ given by 
$$
F\left(A,B\right)=\left|\frac{B}{A}- \sum_{\kappa=1}^{k-1} \frac{A^{\kappa-1}}{\kappa!} \cdot  m^{(\kappa)}(0)\right|\cdot \mathbbm{1}\{h_{n,L}\leq A<h_{n,U}\}.
$$ 
Note that the function class $\mathcal{F}$ is of VC type with a finite VC dimension. 
Moreover, note that 
\begin{align*}
\mathbb{E}[F\left(A,B\right)^2]^{1/2}&=o(1)
\qquad\text{and}
\\
\mathbb{E}\left[\max_{i=1,\ldots,n}F\left(A_i,B_i\right)^2\right]^{1/2}&=o(n^{1/2}),
\end{align*}
because 
\begin{align*}
\mathbb{E}[F\left(A,B\right)^2]^{1/2}
&\leq
\frac{\mathbb{E}\left[B^2\cdot \mathbbm{1}\{h_{n,L}\leq A<h_{n,U}\}\right]^{1/2}}{h_{n,L}}
\\&+ \sum_{\kappa=1}^{k-1}\frac{h_{n,U}^{\kappa-1}}{\kappa!} \cdot  m^{(\kappa)}(0)
\mathbb{E}\left[  \mathbbm{1}\{A<h_{n,U}\}\right]^{1/2}
\\
&=o(1)
\qquad\text{and}
\\
\mathbb{E}\left[\max_{i=1,\ldots,n}F\left(A_i,B_i\right)^2\right]^{1/2}
&\leq
n^{1/2}\mathbb{E}\left[F\left(A_i,B_i\right)^2\right]^{1/2}
\\
&=
\frac{n^{1/2}\mathbb{E}\left[B^2\cdot \mathbbm{1}\{h_{n,L}\leq A<h_{n,U}\}\right]^{1/2}}{h_{n,L}}
\\&+n^{1/2}\sum_{\kappa=1}^{k-1}\frac{h_{n,U}^{\kappa-1}}{\kappa!} \cdot  m^{(\kappa)}(0)\mathbb{E}\left[\mathbbm{1}\{A<h_{n,U}\}\right]^{1/2}
\\
&=
o(n^{1/2})
\end{align*}
under Eq. \eqref{eq:h_{n,U}_small} and \eqref{eq:h_{n,U}_h_{n,L}_veryysmall}. 
Applying Corollary 5.1 of \cite{chernozhukov2014gaussian} with all these characteristics,  
we obtain 
\begin{align*}
&
\mathbb{E}\left[\sup_{\check{h}^1,\check{h}^2\in[h_{n,L},h_{n,U}]}\left|(\mathbb{E}_n-\mathbb{E})\left[\left(\frac{B}{A}- \sum_{\kappa=1}^{k-1} \frac{A^{\kappa-1}}{\kappa!} \cdot  m^{(\kappa)}(0)\right) \cdot \mathbbm{1}\{\check{h}^1\leq A<\check{h}^2\}\right]\right|\right]
\\
&
=
n^{-1/2}O\left(\sqrt{2\mathbb{E}[F\left(A,B\right)^2]}+\frac{2\mathbb{E}\left[\max_{i=1,\ldots,n}F\left(A_i,B_i\right)^2\right]^{1/2}}{\sqrt{n}} \right)
\\
&
=o(n^{-1/2}).
\end{align*}
It completes the proof of this theorem. 
\end{proof}

\subsection{Proof of Corollary \ref{corollary:asymptotic_distribution}: Asymptotic Distribution with Estimated Variance}\label{sec:corollary:asymptotic_distribution}

\begin{proof}
The statement follows from Theorem \ref{theorem:asymptotic_distribution}, Lemma \ref{lemma:variance_estimation}, and Lemma \ref{lemma:estimated_influence_function_in_variance_estimation}.
\end{proof}

\section{Auxiliary Lemmas for the Main Results}\label{sec:auxiliary_lemmas_for_the_main_results}

\subsection{Linear Representation}\label{sec:lemma:linear_representation}

\begin{lemma}[Linear Representation]\label{lemma:linear_representation}
If Assumptions \ref{a:m} (i)--(iii), \ref{a:h} (i), and \ref{a:psi} are satisfied, then
\begin{align*}
\widehat\theta(h_n) - \theta_0 = (\mathbb{E}_n-\mathbb{E})\left[Z(h_n)\right] + o_p(n^{-1/2}).
\end{align*}
\end{lemma}

\begin{proof}
First, we write
\begin{align*}
(\mathbb{E}_n-\mathbb{E})\left[Z(h_n)\right]
&=
\left(\widetilde\theta(h_n) + c(h_n)' \mathbb{E}_n\left[\psi\right]\right)
-
\left(\theta(h_n) + c(h_n)' \mathbb{E}\left[\psi\right]\right)
\end{align*}
by the definitions of $\theta(h_n)$ and $\widetilde\theta(h_n)$.
Rearranging terms and applying Assumption \ref{a:psi} (ii), we in turn obtain
\begin{align*}
\widetilde\theta(h_n) + \left(c(h_n)' \mathbb{E}_n\left[\psi\right] + \theta_0 - \theta(h_n) \right) - \theta_0
=
(\mathbb{E}_n-\mathbb{E})\left[Z(h_n)\right].
\end{align*}
Thus, in order to show the equality claimed in the lemma, it suffices to show
\begin{align}
\label{eq:lemma:lr:1}
c(h_n)' \mathbb{E}_n\left[\psi\right] + \theta_0 - \theta(h_n) = -\widehat \lambda(h_n) + o_p(n^{-1/2}).
\end{align}

By the definition of $\widehat \lambda(h_n)$ given in (\ref{eq:bias_estimator}) and by applying Theorem \ref{lemma:bias} under Assumption \ref{a:m} (i)--(ii), we obtain
\begin{align}
\theta_0 - \theta(h_n) + \widehat \lambda(h_n)
&=
-\sum_{\kappa=1}^{k-1} \frac{\mathbb{E}\left[A^{\kappa-1} \cdot \mathbbm{1}\{0<A<h_n\}\right]}{\kappa!} \cdot \left( \widehat m^{(\kappa)}(0) - m^{(\kappa)}(0) \right)
\label{eq:lemma:lr:2}
\\
&\quad-
\sum_{\kappa=1}^{k-1} \frac{(\mathbb{E}_n-\mathbb{E})\left[A^{\kappa-1} \cdot \mathbbm{1}\{0<A<h_n\}\right]}{\kappa!} \cdot \widehat m^{(\kappa)}(0)
\label{eq:lemma:lr:3}
\\
&\quad+
\frac{\mathbb{E}\left[ A^{k-1} \cdot \mathbbm{1}\{0<A<h_n\} \cdot \int_0^1 (1-t)^{k-1} m^{(k)}(tA)dt \right]}{(k-1)!}.
\label{eq:lemma:lr:4}
\end{align}
The first term (\ref{eq:lemma:lr:2}) can be rewritten as
\begin{align}
&- \sum_{\kappa=1}^{k-1} \frac{\mathbb{E}\left[A^{\kappa-1} \cdot \mathbbm{1}\{0<A<h_n\}\right]}{\kappa!} \cdot \left( \widehat m^{(\kappa)}(0) - m^{(\kappa)}(0) \right)
\nonumber\\
=&
- \sum_{\kappa=1}^{k-1} \frac{\mathbb{E}\left[A^{\kappa-1} \cdot \mathbbm{1}\{0<A<h_n\}\right]}{\kappa!} \cdot \mathbb{E}_n\left[\psi_\kappa\right] + o_p\left(n^{-1/2}\right)
=
- c(h_n)' \mathbb{E}_n\left[\psi\right] + o_p\left( n^{-1/2} \right)
\label{eq:lemma:lr:2prime}
\end{align}
where the first equality is due to Assumption \ref{a:psi} (i), and the second equality is due to the definition of $c(h_n)$ in (\ref{eq:c}).
To evaluate the second term (\ref{eq:lemma:lr:3}), note that 
\begin{align}
(\mathbb{E}_n-\mathbb{E})\left[A^{\kappa-1} \cdot \mathbbm{1}\{0<A<h_n\}\right] = O_p\left(n^{-1/2}h_n^{\kappa-1}\mathbb{E}[\mathbbm{1}\{0<A<h_n\}]\right)
\label{eq:lemma:lr:3prime1}
\end{align}
for each $\kappa \in \{1,...,k-1\}$, and
\begin{align}
\widehat m^{(\kappa)}(0) 
= m^{(\kappa)}(0) + \mathbb{E}_n\left[\psi_\kappa\right] + o_p\left(n^{-1/2}h_n^{-\kappa+1}\right)
= O_p(\max\{1, n^{-1/4}h_n^{-\kappa+1}\})
\label{eq:lemma:lr:3prime2}
\end{align}
for each $\kappa \in \{1,...,k-1\}$ by Assumptions \ref{a:m} (i)--(ii) and \ref{a:psi}.
(Note that Assumption \ref{a:psi} (ii)--(iii) implies $\mathbb{E}_n\left[\psi_\kappa\right] = O_p\left(n^{-1/4}h_n^{-\kappa+1}\right)$.)
Equations (\ref{eq:lemma:lr:3prime1}) and (\ref{eq:lemma:lr:3prime2}) together yield
\begin{align}
-\sum_{\kappa=1}^{k-1} \frac{(\mathbb{E}_n-\mathbb{E})\left[A^{\kappa-1} \cdot \mathbbm{1}\{0<A<h_n\}\right]}{\kappa!} \cdot \widehat m^{(\kappa)}(0) = o_p\left(n^{-1/2}\right)
\label{eq:lemma:lr:3prime}
\end{align}
under Assumptions \ref{a:m} (iii) and \ref{a:h} (i).
(Note that Assumptions \ref{a:m} (iii) and \ref{a:h} (i) imply $\mathbb{E}[\mathbbm{1}\{0<A<h_n\}]=o(1)$ by the dominated convergence theorem.)
The third term (\ref{eq:lemma:lr:4}) is
\begin{eqnarray}
\frac{\mathbb{E}\left[ A^{k-1} \cdot \mathbbm{1}\{0<A<h_n\} \cdot \int_0^1 (1-t)^{k-1} m^{(k)}(tA)dt \right]}{(k-1)!}
&=&
O\left(h_n^{k-1} \cdot \mathbb{E}[\mathbbm{1}\{0<A<h_n\}]\right)
\nonumber\\
&=&
o\left(n^{-1/2}\right),
\label{eq:lemma:lr:4prime}
\end{eqnarray}
where the first equality is due to Assumption \ref{a:m} (i)--(ii), and the second equality is due to Assumptions \ref{a:m} (iii) and \ref{a:h} (i).
Substituting (\ref{eq:lemma:lr:2prime}), (\ref{eq:lemma:lr:3prime}), and (\ref{eq:lemma:lr:4prime}) in (\ref{eq:lemma:lr:2}), (\ref{eq:lemma:lr:3}), and (\ref{eq:lemma:lr:4}), respectively, we obtain (\ref{eq:lemma:lr:1}) as desired.
\end{proof}

\begin{remark}
Note that our assumptions are designed to imply remaining biases of order $o(n^{-1/2})$, e.g., equation (\ref{eq:lemma:lr:4prime}) in the proof of Lemma \ref{lemma:linear_representation} above.
This order of the remaining biases guarantees an asymptotically negligible approximation error for the case of finite variance where the convergence rate is the fastest ever possible.
With this said, we remark that these assumptions can be stronger than necessary in case of infinite variance, as the convergence rate of the trimmed mean may be much slower than the usual $\sqrt{n}$-rate \citep[cf.][]{khan/tamer:2010}.
\end{remark}

\subsection{$L^4$-$L^2$ Ratio}\label{sec:lemma:l4l2}

\begin{lemma}[$L^4$-$L^2$ Ratio]\label{lemma:l4l2}
If Assumptions \ref{a:m} (iii)--(iv), \ref{a:h}, and \ref{a:psi} (iii) are satisfied, then
\begin{align*}
\frac{\mathbb{E}\left[\left(Z(h_n)-\mathbb{E}\left[Z(h_n)\right]\right)^4\right]^{1/4}}{\mathbb{E}\left[\left(Z(h_n)-\mathbb{E}\left[Z(h_n)\right]\right)^2\right]^{1/2}} = o\left(n^{1/4}\right),
\end{align*}
provided that $Var(Z(h_n))$ is bounded away from zero.
\end{lemma}

\begin{proof}
Observe that
\begin{align}
n^{-1/4} \mathbb{E}\left[Z(h_n)^4\right]^{1/4} 
&\leq
n^{-1/4} \mathbb{E}\left[\frac{B^4}{A^4} \cdot \mathbbm{1}\{A \geq h_n\}\right]^{1/4} + n^{-1/4} \mathbb{E}\left[(c(h_n)' \psi)^4\right]^{1/4} 
\nonumber\\
&\leq
n^{-1/4} h_n^{-1} \cdot \mathbb{E}\left[B^4\right]^{1/4} + n^{-1/4} \mathbb{E}\left[(c(h_n)' \psi)^4\right]^{1/4} 
\nonumber\\
&= O\left(n^{-1/4}h_n^{-1}\right) + O\left(\mathbb{E}[\mathbbm{1}\{0<A<h_n\}]\right)
= o\left(1\right),
\label{eq:lemma:l4l2}
\end{align}
where the first inequality is due to Minkowski inequality with the definition of $Z(h_n)$ given in (\ref{eq:z}),
the first equality is due to Assumptions \ref{a:m} (iv) and \ref{a:psi} (iii),
and the last equality is due to Assumptions \ref{a:m} (iii) and \ref{a:h}.
(Note that Assumptions \ref{a:m} (iii) and \ref{a:h} (i) imply $\mathbb{E}[\mathbbm{1}\{0<A<h_n\}]=o(1)$ by the dominated convergence theorem.)
Since
$
\mathbb{E}\left[\left(Z(h_n)-\mathbb{E}\left[Z(h_n)\right]\right)^4\right]
\leq
16 \left( \mathbb{E}\left[Z(h_n)^4\right] + \mathbb{E}\left[Z(h_n)\right]^4 \right),
$
(\ref{eq:lemma:l4l2}) implies
$
n^{-1} \mathbb{E}\left[\left(Z(h_n)-\mathbb{E}\left[Z(h_n)\right]\right)^4\right]
=
o(1).
$
Since $Var(Z(h_n))$ is bounded away from zero, the equality claimed in the lemma follows.
\end{proof}

\subsection{Variance Estimation}\label{sec:lemma:variance_estimation}

\begin{lemma}\label{lemma:variance_estimation}
If Assumptions \ref{a:m} (iii)--(iv), \ref{a:h}, and \ref{a:psi} (iii) are satisfied, then
$$
\left|\frac{\mathbb{E}_n[(Z(h_n)-\mathbb{E}_n[Z(h_n)])^2]}{Var(Z(h_n))}-1\right|
=
o_p(1)
$$
provided that $Var(Z(h_n))$ is bounded away from zero.
\end{lemma}
\begin{proof}
First, note that we have
\begin{align*}
Var\left((\mathbb{E}_n-\mathbb{E})\left[\frac{(Z(h_n)-\mathbb{E}[Z(h_n)])^2}{Var(Z(h_n))}\right]\right)
=&
\frac{Var((Z(h_n)-\mathbb{E}[Z(h_n)])^2)}{n\cdot Var(Z(h_n))^2}\\
=&
\frac{\mathbb{E}[(Z(h_n)-\mathbb{E}[Z(h_n)])^4]-\mathbb{E}[(Z(h_n)-\mathbb{E}[Z(h_n)])^2]^2}{n\cdot Var(Z(h_n))^2}\\
=&
\frac{\mathbb{E}[(Z(h_n)-\mathbb{E}[Z(h_n)])^4]}{n \cdot \mathbb{E}[(Z(h_n)-\mathbb{E}[Z(h_n)])^2]^2}-\frac{1}{n}
=
o(1),
\end{align*}
where the last equality is due to Lemma \ref{lemma:l4l2} under Assumptions \ref{a:m} (iii)--(iv), \ref{a:h}, and \ref{a:psi} (iii), since $Var(Z(h_n))$ is bounded away from zero.
Therefore,
\begin{align}
(\mathbb{E}_n-\mathbb{E})\left[\frac{(Z(h_n)-\mathbb{E}[Z(h_n)])^2}{Var(Z(h_n))}\right]
=
o_p(1).
\label{eq:lemma:variance_estimation1}
\end{align}
Second, note that we have
\begin{align*}
Var\left(\mathbb{E}_n\left[\frac{Z(h_n)-\mathbb{E}[Z(h_n)]}{\sqrt{Var(Z(h_n))}}\right]\right)
=
\frac{Var\left(Z(h_n)-\mathbb{E}[Z(h_n)]\right)}{n \cdot Var(Z(h_n))}
=
\frac{1}{n}.
\end{align*}
Therefore,
\begin{align}
\mathbb{E}_n\left[\frac{Z(h_n)-\mathbb{E}[Z(h_n)]}{\sqrt{Var(Z(h_n))}}\right]
=
O_p\left(n^{-1/2}\right).
\label{eq:lemma:variance_estimation2}
\end{align}
Third, we have
\begin{align}
\mathbb{E}_n[(Z(h_n)-\mathbb{E}_n[Z(h_n)])^2]
=&
\mathbb{E}_n[(Z(h_n)-\mathbb{E}[Z(h_n)]+\mathbb{E}[Z(h_n)]-\mathbb{E}_n[Z(h_n)])^2]
\notag\\
=&
\mathbb{E}_n[(Z(h_n)-\mathbb{E}[Z(h_n)]-\mathbb{E}_n[Z(h_n)-\mathbb{E}[Z(h_n)]])^2]
\notag\\
=&
\mathbb{E}_n[(Z(h_n)-\mathbb{E}[Z(h_n)])^2]-\mathbb{E}_n[Z(h_n)-\mathbb{E}[Z(h_n)]]^2.
\label{eq:lemma:variance_estimation3}
\end{align}
Combining the three auxiliary results above, we obtain
\begin{align*}
\frac{\mathbb{E}_n[(Z(h_n)-\mathbb{E}_n[Z(h_n)])^2]}{Var(Z(h_n))}-1
=&
\frac{(\mathbb{E}_n-\mathbb{E})[(Z(h_n)-\mathbb{E}[Z(h_n)])^2]-\mathbb{E}_n[Z(h_n)-\mathbb{E}[Z(h_n)]]^2}{Var(Z(h_n))}\\
=&
(\mathbb{E}_n-\mathbb{E})\left[\frac{(Z(h_n)-\mathbb{E}[Z(h_n)])^2}{Var(Z(h_n))}\right]-\mathbb{E}_n\left[\frac{Z(h_n)-\mathbb{E}[Z(h_n)]}{\sqrt{Var(Z(h_n))}}\right]^2\\
=& o_p(1),
\end{align*}
where the first equality is due to (\ref{eq:lemma:variance_estimation3}),
and the last equality is due to (\ref{eq:lemma:variance_estimation1}) and (\ref{eq:lemma:variance_estimation2}).
The equality claimed in the lemma follows.
\end{proof}

\subsection{Estimated Influence Function in Variance Estimation}\label{sec:lemma:estimated_influence_function_in_variance_estimation}

\begin{lemma}\label{lemma:estimated_influence_function_in_variance_estimation}
If Assumption \ref{a:psi}$'$ is satisfied, 
then 
$$
\left|\frac{\mathbb{E}_n\left[\widehat{Z}(h_n)^2\right]-\mathbb{E}_n\left[\widehat{Z}(h_n)\right]^2}{\mathbb{E}_n\left[Z(h_n)^2\right]-\mathbb{E}_n\left[Z(h_n)\right]^2}-1\right|
=
o_p(1)
$$
provided that $Var(Z(h_n))$ is bounded away from zero.
\end{lemma}
\begin{proof}
Since $0\leq A^{\kappa-1} \cdot \mathbbm{1}\{0<A<h_n\}\leq 1$ for each $\kappa \in \{0,...,k-2\}$ for sufficiently small $h_n$, 
we have 
\begin{align}
\|\widehat{c}(h_n)\|= 
\left\|\left(
\begin{array}{c}
\mathbb{E}_n\left[\mathbbm{1}\{0<A<h_n\}\right]/1!\\
\mathbb{E}_n\left[A \cdot \mathbbm{1}\{0<A<h_n\}\right]/2!\\
\vdots\\
\mathbb{E}_n\left[A^{k-2} \cdot \mathbbm{1}\{0<A<h_n\}\right]/(k-1)!\\
\end{array}
\right)
\right\|
=
O_p(1)
\label{eq:lemma:estimated_influence_function_in_variance_estimation1}
\end{align}
and 
\begin{align}
\widehat{c}(h_n)-c(h_n) = 
\left(
\begin{array}{c}
(\mathbb{E}_n-\mathbb{E})\left[\mathbbm{1}\{0<A<h_n\}\right]/1!\\
(\mathbb{E}_n-\mathbb{E})\left[A \cdot \mathbbm{1}\{0<A<h_n\}\right]/2!\\
\vdots\\
(\mathbb{E}_n-\mathbb{E})\left[A^{k-2} \cdot \mathbbm{1}\{0<A<h_n\}\right]/(k-1)!\\
\end{array}
\right)
=
o_p(n^{-1/2}).
\label{eq:lemma:estimated_influence_function_in_variance_estimation2}
\end{align}
Therefore, from (\ref{eq:z}) and (\ref{eq:Z(h_n)at}), we have
\begin{align}
\mathbb{E}_n\left[\left(\widehat{Z}(h_n)-Z(h_n)\right)^2\right]^{1/2}
&=
\mathbb{E}_n\left[\left(\widehat{c}(h_n)' \widehat{\psi}-c(h_n)' \psi\right)^2\right]^{1/2}
\notag\\
&=
\mathbb{E}_n\left[\left(\widehat{c}(h_n)' \left(\widehat{\psi}-\psi\right)+\left(\widehat{c}(h_n)-c(h_n)\right)'\psi\right)^2\right]^{1/2}
\notag\\
&\leq
\mathbb{E}_n\left[\left(\widehat{c}(h_n)' \left(\widehat{\psi}-\psi\right)\right)^2\right]^{1/2}+\mathbb{E}_n\left[\left(\left(\widehat{c}(h_n)-c(h_n)\right)'\psi\right)^2\right]^{1/2}
\notag\\
&\leq
\mathbb{E}_n\left[\norm{\widehat{c}(h_n)}^2\cdot\norm{\widehat{\psi}-\psi}^2\right]^{1/2}+\mathbb{E}_n\left[\norm{\widehat{c}(h_n)-c(h_n)}^2\norm{\psi}^2\right]^{1/2}
\notag\\
&=
\norm{\widehat{c}(h_n)}\cdot \mathbb{E}_n\left[\norm{\widehat{\psi}-\psi}^2\right]^{1/2}+\norm{\widehat{c}(h_n)-c(h_n)}\cdot \mathbb{E}_n\left[\norm{\psi}^2\right]^{1/2}
\notag\\
&=o_p(1),
\label{eq:lemma:estimated_influence_function_in_variance_estimation3}
\end{align}
where 
the inequalities are due to Minkowski's and Cauchy-Schwarz inequalities, and
the last equality is due to (\ref{eq:lemma:estimated_influence_function_in_variance_estimation1}), (\ref{eq:lemma:estimated_influence_function_in_variance_estimation2}), and Assumption \ref{a:psi}$'$.

Now, observe that
\begin{align*}
&
\mathbb{E}_n\left[\widehat{Z}(h_n)^2\right]-\mathbb{E}_n\left[\widehat{Z}(h_n)\right]^2
\\
&=
\mathbb{E}_n\left[\left(\widehat{Z}(h_n)-Z(h_n)+Z(h_n)\right)^2\right]-\mathbb{E}_n\left[\widehat{Z}(h_n)-Z(h_n)+Z(h_n)\right]^2\\
&=
\mathbb{E}_n\left[Z(h_n)^2\right]-\mathbb{E}_n\left[Z(h_n)\right]^2
+\mathbb{E}_n\left[\left(\widehat{Z}(h_n)-Z(h_n)\right)^2\right]\\&\quad+2\mathbb{E}_n\left[\left(\widehat{Z}(h_n)-Z(h_n)\right)Z(h_n)\right]\\
&\quad-\mathbb{E}_n\left[\widehat{Z}(h_n)-Z(h_n)\right]^2-2\mathbb{E}_n\left[\widehat{Z}(h_n)-Z(h_n)\right] \cdot \mathbb{E}_n\left[Z(h_n)\right]\\
&=
\mathbb{E}_n\left[Z(h_n)^2\right]-\mathbb{E}_n\left[Z(h_n)\right]^2
+\mathbb{E}_n\left[\left(\widehat{Z}(h_n)-Z(h_n)\right)^2\right]\\&\quad+2\mathbb{E}_n\left[\left(\widehat{Z}(h_n)-Z(h_n)\right)\cdot\left(Z(h_n)-\mathbb{E}\left[Z(h_n)\right]\right)\right]
\\&-\mathbb{E}_n\left[\widehat{Z}(h_n)-Z(h_n)\right]^2-2\mathbb{E}_n\left[\widehat{Z}(h_n)-Z(h_n)\right] \cdot \mathbb{E}_n\left[Z(h_n)-\mathbb{E}\left[Z(h_n)\right]\right].
\end{align*}
Therefore, we obtain
\begin{align*}
&
\left| \frac{\mathbb{E}_n\left[\widehat{Z}(h_n)^2\right]-\mathbb{E}_n\left[\widehat{Z}(h_n)\right]^2}{\mathbb{E}_n\left[Z(h_n)^2\right]-\mathbb{E}_n\left[Z(h_n)\right]^2}-1 \right|
\\
&\quad\leq
\frac{\mathbb{E}_n\left[\left(\widehat{Z}(h_n)-Z(h_n)\right)^2\right]+2\left|\mathbb{E}_n\left[\left(\widehat{Z}(h_n)-Z(h_n)\right) \cdot \left(Z(h_n)-\mathbb{E}\left[Z(h_n)\right]\right)\right]\right|}{\mathbb{E}_n\left[Z(h_n)^2\right]-\mathbb{E}_n\left[Z(h_n)\right]^2}\\
&\qquad+\frac{\mathbb{E}_n\left[\widehat{Z}(h_n)-Z(h_n)\right]^2+2\left|\mathbb{E}_n\left[\widehat{Z}(h_n)-Z(h_n)\right]\right|\cdot \left|\mathbb{E}_n\left[\left(Z(h_n)-\mathbb{E}\left[Z(h_n)\right]\right)\right]\right|}{\mathbb{E}_n\left[Z(h_n)^2\right]-\mathbb{E}_n\left[Z(h_n)\right]^2}\\
&\quad\leq
\frac{\mathbb{E}_n\left[\left(\widehat{Z}(h_n)-Z(h_n)\right)^2\right]+2\sqrt{\mathbb{E}_n\left[\left(\widehat{Z}(h_n)-Z(h_n)\right)^2\right] \cdot \mathbb{E}_n\left[\left(Z(h_n)-\mathbb{E}\left[Z(h_n)\right]\right)^2\right]}}{\mathbb{E}_n\left[Z(h_n)^2\right]-\mathbb{E}_n\left[Z(h_n)\right]^2}\\
&\qquad+\frac{\mathbb{E}_n\left[\widehat{Z}(h_n)-Z(h_n)\right]^2+2\sqrt{\mathbb{E}_n\left[\left(\widehat{Z}(h_n)-Z(h_n)\right)^2\right]\cdot \mathbb{E}_n\left[\left(Z(h_n)-\mathbb{E}\left[Z(h_n)\right]\right)^2\right]}}{\mathbb{E}_n\left[Z(h_n)^2\right]-\mathbb{E}_n\left[Z(h_n)\right]^2}\\
&\quad=
\frac{\mathbb{E}_n\left[\left(\widehat{Z}(h_n)-Z(h_n)\right)^2\right]}{\mathbb{E}_n\left[Z(h_n)^2\right]-\mathbb{E}_n\left[Z(h_n)\right]^2}
+\frac{\mathbb{E}_n\left[\widehat{Z}(h_n)-Z(h_n)\right]^2}{\mathbb{E}_n\left[Z(h_n)^2\right]-\mathbb{E}_n\left[Z(h_n)\right]^2}
+4\sqrt{\frac{\mathbb{E}_n\left[\left(\widehat{Z}(h_n)-Z(h_n)\right)^2\right]}{\mathbb{E}_n\left[Z(h_n)^2\right]-\mathbb{E}_n\left[Z(h_n)\right]^2}}
\\
&\quad=
o_p(1),
\end{align*}
where
the inequalities are due to triangle and Cauchy-Schwarz inequalities, and
the last equality is due to (\ref{eq:lemma:estimated_influence_function_in_variance_estimation3}), and the assumption that the denominator is bounded away from zero.
\end{proof}

\section{Proofs of the Extended Results}\label{sec:proofs_of_the_extended_results}
\subsection{Proof of Theorem \ref{lemma:ex:bias}: Bias Characterization}\label{sec:lemma:ex:bias}
\begin{proof}
For $h_n$ small enough, the Taylor expansion of $m$ around 0 yields
\begin{align}\label{eq:ex:taylor}
m(a) 
&= 
\sum_{\kappa=1}^{k-1} \frac{a^\kappa}{\kappa!} \cdot m^{(\kappa)}(0) + \frac{a^k}{(k-1)!} \cdot \int_0^1 (1-t)^{k-1} m^{(k)}(ta)dt
\end{align}
under Assumption \ref{a:ex:m} (i)--(ii).
We can write the bias $\theta(h_n)-\theta_0$ as
\begin{align*}
&\theta(h_n) - \theta_0
\\
&=
\mathbb{E}\left[\frac{B \cdot \left(S\left(\frac{A}{h_n}\right) - 1\right)}{A} \right]
\\
&=
\mathbb{E}\left[\frac{m(A) \cdot \left(S\left(\frac{A}{h_n}\right) - 1\right)}{A} \right]
\\
&=
\mathbb{E}\left[  \sum_{\kappa=1}^{k-1} \frac{A^{\kappa-1}}{\kappa!} \cdot m^{(\kappa)}(0) \left(S\left(\frac{A}{h_n}\right) - 1\right) 
+  \frac{A^{k-1}}{(k-1)!} \cdot \int_0^1 (1-t)^{k-1} m^{(k)}(tA)dt \left(S\left(\frac{A}{h_n}\right)-1\right)  \right]
\\
&=
\sum_{\kappa=1}^{k-1} \frac{\mathbb{E}\left[A^{\kappa-1} \left(S\left(\frac{A}{h_n}\right)-1\right)\right]}{\kappa !} \cdot m^{(\kappa)}(0)
\\
&\quad
+
\frac{\mathbb{E}\left[ A^{k-1} \cdot \int_0^1 (1-t)^{k-1} m^{(k)} (tA) dt \left(S\left(\frac{A}{h_n}\right)-1\right)\right]}{(k-1)!},
\end{align*}
where 
the first equality follows from the definitions of $\theta(h_n)$ and $\theta_0$ in (\ref{eq:ex:theta_h}),
the second equality follows from the law of iterated expectations, and
the third equality follows from (\ref{eq:ex:taylor}).
The second term in the last expression can be in turn rewritten as
$$
\frac{\mathbb{E}\left[A^{k-1} \cdot \int_0^1 (1-t)^{k-1} m^{(k)} (tA) dt \left(S\left(\frac{A}{h_n}\right)-1\right) \right]}{(k-1)!}
=
O\left(h_n^{k-1} \cdot \mathbb{E}[\mathbbm{1}\{0<A<h_n\}]\right)
=
o\left(n^{-1/2}\right),
$$
where 
the second equality follows from Assumptions \ref{a:ex:m} (i)--(ii) and \ref{a:ex:s} (i), and 
the last equality follows from Assumption \ref{a:ex:h} (i).
Therefore, the claimed equality follows.
\end{proof}

\subsection{Proof of Theorem \ref{theorem:ex:asymptotic_distribution}: Asymptotic Distribution}\label{sec:theorem:ex:asymptotic_distribution}
\begin{proof}
By Lyapunov's central limit theorem and (\ref{eq:ex:l4l2}), we have
$$
\frac{ (\mathbb{E}_n-\mathbb{E}) \left[Z(h_n)\right] }{\sqrt{Var(Z(h_n))/n}}
\stackrel{d}{\rightarrow} \mathcal{N}(0,1).
$$
Also, observe that
\begin{align}
\frac{ (\mathbb{E}_n-\mathbb{E}) \left[Z(h_n)\right] }{\sqrt{Var(Z(h_n))/n}}
&=
\frac{ \widehat\theta(h_n) - \theta_0 }{\sqrt{Var(Z(h_n))/n}}
+
\frac{ o_p\left(1\right) }{\sqrt{Var(Z(h_n))}}
=
\frac{ \widehat\theta(h_n) - \theta_0 }{\sqrt{Var(Z(h_n))/n}}
+
o_p(1)
\end{align}
by Theorem \ref{lemma:ex:bias} and Lemma \ref{lemma:ex:influence_function}
under Assumptions \ref{a:ex:m}, \ref{a:ex:s}, \ref{a:ex:h}, \ref{a:ex:psi}, and \ref{a:ex:gamma}.
Applying Slutsky's theorem yields the result claimed in the statement of the theorem.
\end{proof}

\section{Auxiliary Lemmas for the Extended Results}\label{sec:auxiliary_lemmas_for_the_extended_results}

\subsection{Bounds for $\omega_{1,n}$}

\begin{lemma}[Bounds for $\omega_{1,n}$]\label{lemma:ex:omega1}
If Assumptions \ref{a:ex:m} (iii) and (v), \ref{a:ex:s} (ii)--(iii), and \ref{a:ex:h} (ii) are satisfied, then
\begin{align}
\mathbb{E}\left[\norm{\left.\frac{\partial}{\partial\gamma} \omega_{1,n}(X,\gamma)\right|_{\gamma=\gamma_0}}^2\right]
&=
o(n)
\qquad\text{and}
\label{eq:ex:omega1_1}
\\
\mathbb{E}_n\left[\sup_{\gamma \in \Gamma} \norm{\frac{\partial^2}{\partial\gamma\partial\gamma'} \omega_{1,n}(X,\gamma)} \right]
&=
o_p\left(n^{1/2}\right).
\label{eq:ex:omega1_2}
\end{align}
\end{lemma}
\begin{proof}
First, note that
\begin{align*}
\frac{\partial}{\partial\gamma} \omega_{1,n}(X,\gamma)
&=
h_n^{-1}\frac{S\left(\frac{g_A(X,\gamma)}{h_n}\right)}{g_A(X,\gamma)/h}
\frac{\partial}{\partial\gamma} g_B(X,\gamma)
+h_n^{-2}  \cdot \left.\frac{d}{du}\frac{S(u)}{u} \right\vert_{u=g_A(X,\gamma)/h}\cdot g_B(X,\gamma)
\cdot \frac{\partial}{\partial\gamma} g_A(X,\gamma)
\end{align*}
under Assumptions \ref{a:ex:m} (iii) and \ref{a:ex:s} (ii). 
Thus, we obtain
\begin{align*}
\mathbb{E}\left[\norm{\left.\frac{\partial}{\partial\gamma} \omega_{1,n}(X,\gamma)\right|_{\gamma=\gamma_0}}^2\right]^{1/2}
&\leq
h_n^{-1} \cdot \sup_{u} \abs{\frac{S(u)}{u}} \cdot \mathbb{E}\left[\norm{\left.\frac{\partial}{\partial \gamma} g_B(X,\gamma)\right\vert_{\gamma=\gamma_0}}^2\right]^{1/2}
\\&\quad+
h_n^{-2} \cdot \sup_{u} \abs{\frac{d}{du} \frac{S(u)}{u}} \cdot \mathbb{E}\left[g_B(X,\gamma_0)^2 \cdot  \norm{\left.\frac{\partial}{\partial \gamma} g_A(X,\gamma)\right\vert_{\gamma=\gamma_0}}^2\right]^{1/2}
\\&=
o(\sqrt{n}),
\end{align*}
where 
the last equality follows from Assumptions \ref{a:ex:m} (v), \ref{a:ex:s} (iii), and \ref{a:ex:h} (ii).
This establishes (\ref{eq:ex:omega1_1}).

Second, note that
\begin{align*}
\frac{\partial^2}{\partial\gamma\partial\gamma'}\omega_{1,n}(X,\gamma)
&=
2h_n^{-2}
\frac{\partial}{\partial\gamma'} g_A(X,\gamma)\frac{\partial}{\partial\gamma} g_B(X,\gamma)\left.\frac{d}{du}\frac{S\left(u\right)}{u}\right|_{u=g_A(X,\gamma)/h}
\\&+
h_n^{-1}\frac{\partial^2}{\partial\gamma\partial\gamma'} g_B(X,\gamma)\frac{S\left(\frac{g_A(X,\gamma)}{h_n}\right)}{g_A(X,\gamma)/h}
\\&+
h_n^{-3} g_B(X,\gamma)\cdot \frac{\partial^2}{\partial\gamma\partial\gamma'} g_A(X,\gamma) \cdot \left.\frac{d}{du}\frac{S(u)}{u} \right\vert_{u=g_A(X,\gamma)/h}
\\&+
h_n^{-3} g_B(X,\gamma)\cdot \frac{\partial}{\partial\gamma} g_A(X,\gamma)\frac{\partial}{\partial\gamma'} g_A(X,\gamma) \cdot \left.\frac{d^2}{du^2}\frac{S(u)}{u} \right\vert_{u=g_A(X,\gamma)/h}
\end{align*}
under Assumption \ref{a:ex:s} (ii). 
Thus, we obtain
\begin{align*}
\mathbb{E}_n\left[\sup_{\gamma \in \Gamma} \norm{\frac{\partial^2}{\partial\gamma\partial\gamma'}\omega_{1,n}(X,\gamma)} \right]
&\leq
2h_n^{-2}
\mathbb{E}_n\left[\sup_{\gamma \in \Gamma}\norm{\frac{\partial}{\partial\gamma'} g_A(X,\gamma)}\norm{\frac{\partial}{\partial\gamma} g_B(X,\gamma)}\right]\sup_u \abs{ \frac{d}{du}\frac{S(u)}{u} } 
\\&+
h_n^{-1}\mathbb{E}_n\left[\sup_{\gamma \in \Gamma}\norm{
\frac{\partial^2}{\partial\gamma\partial\gamma'} g_B(X,\gamma)}\right]\sup_u\left|\frac{S\left(u\right)}{u}\right|
\\&+
h_n^{-3}\mathbb{E}_n\left[ \sup_{\gamma \in \Gamma}|g_B(X,\gamma)|\cdot \norm{\frac{\partial^2}{\partial\gamma\partial\gamma'} g_A(X,\gamma)} \right]\cdot \sup_u \abs{ \frac{d}{du}\frac{S(u)}{u} } 
\\&+
h_n^{-3} \mathbb{E}_n\left[\sup_{\gamma \in \Gamma}|g_B(X,\gamma)|\cdot \norm{\frac{\partial}{\partial\gamma} g_A(X,\gamma)}^2\right] \cdot \sup_u \abs{ \frac{d^2}{du^2}\frac{S(u)}{u} } 
\\
&=
o_p\left(n^{1/2}\right),
\end{align*}
where 
the last equality follows from Assumptions \ref{a:ex:m} (v), \ref{a:ex:s} (iii), and \ref{a:ex:h} (ii), as well as the weak law of large numbers.
This establishes (\ref{eq:ex:omega1_2}).
\end{proof}

\subsection{Bounds for $\omega_{2,\kappa,n}$}
\begin{lemma}[Bounds for $\omega_{2,\kappa,n}$]\label{lemma:ex:omega2}
If Assumptions \ref{a:ex:m} (iii)--(iv), \ref{a:ex:s} (i)--(iv), and \ref{a:ex:h} (ii) are satisfied, then
\begin{align}
\mathbb{E}\left[\norm{\left.\frac{\partial}{\partial\gamma} \omega_{2,\kappa,n}(X,\gamma)\right|_{\gamma=\gamma_0}}^2\right]
&=
o(n)
\qquad\text{and}
\label{eq:ex:omega2_1}
\\
\mathbb{E}_n\left[\sup_{\gamma \in \Gamma} \norm{\frac{\partial^2}{\partial\gamma\partial\gamma'} \omega_{2,\kappa,n}(X,\gamma)}\right]
&=
o_p\left(n^{1/2}\right).
\label{eq:ex:omega2_2}
\end{align}
\end{lemma}
\begin{proof}
First, note that
\begin{align*}
\frac{\partial}{\partial\gamma} \omega_{2,\kappa,n}(X,\gamma)
=
\frac{\partial}{\partial\gamma} g_A(X,\gamma) \cdot \left((\kappa-1) \cdot g_A(X,\gamma)^{\kappa-2} \cdot \left(S\left(\frac{g_A(X,\gamma)}{h_n}\right)-1\right) 
\right. \
\\
+
\left. 
h_n^{-1} \cdot g_A(X,\gamma)^{\kappa-1} \cdot S'\left(\frac{g_A(X,\gamma)}{h_n}\right) \right)
\end{align*}
under Assumptions \ref{a:ex:m} (iii) and \ref{a:ex:s} (ii).
Thus,
\begin{align*}
\mathbb{E}\left[\norm{\left.\frac{\partial}{\partial\gamma} \omega_{2,\kappa,n}(X,\gamma)\right|_{\gamma=\gamma_0}}^2\right]
&\lesssim
\sup_x \abs{g_A(x,\gamma_0)}^{2\kappa-4} \cdot \left(\sup_u\abs{S(u)}+1\right)^2 \cdot \mathbb{E}\left[\norm{\left.\frac{\partial}{\partial \gamma}g_A(X,\gamma)\right|_{\gamma=\gamma_0}}^2\right]
\\
&+
h_n^{-2} \cdot \sup_x \abs{g_A(x,\gamma_0)}^{2\kappa-2} \cdot \sup_u\abs{S'(u)}^2 \cdot \mathbb{E}\left[\norm{\left.\frac{\partial}{\partial \gamma}g_A(X,\gamma)\right|_{\gamma=\gamma_0}}^2\right]
\\
&=
o(n),
\end{align*}
where the last equality follows from Assumptions \ref{a:ex:m} (iv), \ref{a:ex:s} (iv), and \ref{a:ex:h} (ii).
This establishes (\ref{eq:ex:omega2_1}).

Second, note that
\begin{align*}
\frac{\partial^2}{\partial\gamma\partial\gamma'} \omega_{2,\kappa,n}(X,\gamma)
&=
\frac{\partial}{\partial\gamma} g_A(X,\gamma) 
\cdot 
\left[ (\kappa-1) \cdot (\kappa-2) \cdot g_A(X,\gamma)^{\kappa-3} \cdot \left(S\left(\frac{g_A(X,\gamma)}{h_n}\right)-1\right) \right. \
\\ &
\qquad +
2 \cdot h_n^{-1} \cdot (\kappa-1) \cdot g_A(X,\gamma)^{\kappa-2} \cdot S'\left(\frac{g_A(X,\gamma)}{h_n}\right) \
\\ &
\qquad +
\left. h_n^{-2} \cdot g_A(X,\gamma)^{\kappa-1} \cdot S''\left(\frac{g_A(X,\gamma)}{h_n}\right) \right]
\cdot
\frac{\partial}{\partial\gamma'} g_A(X,\gamma) 
\\ &
\quad +
\left[(\kappa-1) \cdot g_A(X,\gamma)^{\kappa-2} \cdot \left(S\left(\frac{g_A(X,\gamma)}{h_n}\right)-1\right) \right. \qquad
\\ &
\qquad +
\left. h_n^{-1} \cdot g_A(X,\gamma)^{\kappa-1} \cdot S'\left(\frac{g_A(X,\gamma)}{h_n}\right) \right]
\cdot \frac{\partial^2}{\partial\gamma\partial\gamma'} g_A(X,\gamma)
\end{align*}
under Assumption \ref{a:ex:s} (ii).
Thus, we obtain
\begin{align*}
&\mathbb{E}_n\left[\sup_{\gamma \in \Gamma} \norm{\frac{\partial^2}{\partial\gamma\partial\gamma'} \omega_{2,\kappa,n}(X,\gamma)}\right]
\\
&\quad\lesssim
\underset{\gamma\in\Gamma}{\sup} \ \underset{x}{\sup} \abs{g_A(x,\gamma)}^{\kappa-3} \cdot \left(\sup_u \abs{S(u)}+1\right) \cdot \mathbb{E}_n\left[\underset{\gamma\in\Gamma}{\sup}\norm{\frac{\partial}{\partial\gamma}g_A(X,\gamma)}^2\right]
\\
&\qquad\qquad+
h_n^{-1} \cdot
\underset{\gamma\in\Gamma}{\sup} \ \underset{x}{\sup} \abs{g_A(x,\gamma)}^{\kappa-2} \cdot \sup_u \abs{S'(u)} \cdot \mathbb{E}_n\left[\underset{\gamma\in\Gamma}{\sup}\norm{\frac{\partial}{\partial\gamma}g_A(X,\gamma)}^2\right]
\\
&\qquad\qquad+
h_n^{-2} \cdot
\underset{\gamma\in\Gamma}{\sup} \ \underset{x}{\sup} \abs{g_A(x,\gamma)}^{\kappa-1} \cdot \sup_u \abs{S''(u)} \cdot \mathbb{E}_n\left[\underset{\gamma\in\Gamma}{\sup}\norm{\frac{\partial}{\partial\gamma}g_A(X,\gamma)}^2\right]
\\
&\qquad\qquad+
\underset{\gamma\in\Gamma}{\sup} \ \underset{x}{\sup} \abs{g_A(x,\gamma)}^{\kappa-2} \cdot \left( \sup_u \abs{S(u)} + 1 \right) \cdot \mathbb{E}_n\left[\underset{\gamma\in\Gamma}{\sup}\norm{\frac{\partial^2}{\partial\gamma\partial\gamma'}g_A(X,\gamma)}\right]
\\
&\qquad\qquad+
h_n^{-1} \cdot
\underset{\gamma\in\Gamma}{\sup} \ \underset{x}{\sup} \abs{g_A(x,\gamma)}^{\kappa-1} \cdot \sup_u \abs{S'(u)} \cdot \mathbb{E}_n\left[\underset{\gamma\in\Gamma}{\sup}\norm{\frac{\partial^2}{\partial\gamma\partial\gamma'}g_A(X,\gamma)}\right]
\\&\quad=
o_p\left(n^{1/2}\right),
\end{align*}
where the last equality follows from Assumptions \ref{a:ex:m} (iv), \ref{a:ex:s} (i) and (iv), and \ref{a:ex:h} (ii), as well as the weak law of large numbers.
This establishes (\ref{eq:ex:omega2_2}).
\end{proof}

\subsection{Trimmed Mean Estimator}

\begin{lemma}\label{lemma:ex:trimmed_mean}
If Assumptions \ref{a:ex:m} (iii) and (v), \ref{a:ex:s} (ii)--(iii), \ref{a:ex:h} (ii), and \ref{a:ex:gamma} are satisfied, then
\begin{align*}
&\mathbb{E}_n\left[\frac{g_B(X,\widehat\gamma)}{g_A(X,\widehat\gamma)} \cdot S\left(\frac{g_A(X,\widehat\gamma)}{h_n}\right)\right]
-
\mathbb{E}_n\left[\frac{g_B(X,\gamma_0)}{g_A(X,\gamma_0)} \cdot S\left(\frac{g_A(X,\gamma_0)}{h_n}\right)\right]
\\
=&
\mathbb{E}_n \left[ \mathbb{E}\left[ \left.\frac{\partial}{\partial\gamma'}\omega_{1,n}(X,\gamma)\right|_{\gamma=\gamma_0} \right] \cdot \phi \right]
+ o_p\left(n^{-1/2}\right).
\end{align*}
\end{lemma}
\begin{proof}
We present four sets of auxiliary calculations.
First,
\begin{align}
\omega_{1,n}(x,\gamma) - \omega_{1,n}(x,\gamma_0)
&=
\left.\frac{\partial}{\partial\gamma'} \omega_{1,n}(X,\gamma)\right|_{\gamma=\gamma_0} \left(\gamma - \gamma_0\right)
+
R_{1,h}(x,\gamma,\gamma_0)
\qquad\text{where}
\label{eq:ex:taylor1}
\\
\abs{R_{1,h}(x,\gamma,\gamma_0)}
& \leq
\sup_{\gamma\in\Gamma} \norm{\frac{\partial^2}{\partial\gamma\partial\gamma'} \omega_{1,n}(x,\gamma)} \cdot \norm{\gamma-\gamma_0}^2
\label{eq:ex:taylor1_remainder}
\end{align}
by Taylor's theorem under Assumptions \ref{a:ex:m} (iii) and \ref{a:ex:s} (ii).
Second,
\begin{align}
\abs{\mathbb{E}_n\left[R_{1,h}(X,\widehat\gamma,\gamma_0)\right]}
\leq
\mathbb{E}_n\left[\sup_{\gamma\in\Gamma} \norm{\frac{\partial^2}{\partial\gamma\partial\gamma'} \omega_{1,n}(X,\gamma)} \right] \cdot \norm{\widehat\gamma - \gamma_0}^2
=
o_p\left(n^{-1/2}\right)
\label{eq:lemma:ex:trimmed_mean1}
\end{align}
where the inequality is due to (\ref{eq:ex:taylor1_remainder}) 
and 
the equality is due to Assumption \ref{a:ex:gamma} and (\ref{eq:ex:omega1_2}) of Lemma \ref{lemma:ex:omega1} under Assumptions \ref{a:ex:m} (iii) and (v), \ref{a:ex:s} (ii)--(iii), and \ref{a:ex:h} (ii). 
Third,
\begin{align}
\mathbb{E}\left[\left.\frac{\partial}{\partial\gamma'} \omega_{1,n}(X,\gamma)\right|_{\gamma=\gamma_0} \right] \cdot \left(\widehat\gamma - \gamma_0\right) 
=
\mathbb{E}\left[\left.\frac{\partial}{\partial\gamma}\omega_{1,n}(X,\gamma)\right|_{\gamma=\gamma_0}\right]'\cdot \mathbb{E}_n[\phi]
+
o_p\left(n^{-1/2}\right),
\label{eq:lemma:ex:trimmed_mean2}
\end{align}
where the equality is due to Assumption \ref{a:ex:gamma} and (\ref{eq:ex:omega1_1}) of Lemma \ref{lemma:ex:omega1} under Assumptions \ref{a:ex:m} (iii) and (v), \ref{a:ex:s} (ii)--(iii), and \ref{a:ex:h} (ii).
Fourth,
\begin{align}
\left(\mathbb{E}_n-\mathbb{E}\right)\left[\left.\frac{\partial}{\partial\gamma'} \omega_{1,n}(X,\gamma)\right|_{\gamma=\gamma_0} \right] \cdot \left(\widehat\gamma - \gamma_0\right)
=
o_p\left(n^{-1/2}\right),
\label{eq:lemma:ex:trimmed_mean3}
\end{align}
where the equality is due to Assumption \ref{a:ex:gamma} and (\ref{eq:ex:omega1_1}) of Lemma \ref{lemma:ex:omega1} under Assumptions \ref{a:ex:m} (iii) and (v), \ref{a:ex:s} (ii)--(iii), and \ref{a:ex:h} (ii).

We now obtain
\begin{align*}
&\mathbb{E}_n\left[\frac{g_B(X,\widehat\gamma)}{g_A(X,\widehat\gamma)} \cdot S\left(\frac{g_A(X,\widehat\gamma)}{h_n}\right)\right]
-
\mathbb{E}_n\left[\frac{g_B(X,\gamma_0)}{g_A(X,\gamma_0)} \cdot S\left(\frac{g_A(X,\gamma_0)}{h_n}\right)\right]
\\
=&
\mathbb{E}_n\left[\omega_{1,n}(X,\widehat\gamma) - \omega_{1,n}(X,\gamma_0)\right]
\\
=&
\mathbb{E}_n\left[\left.\frac{\partial}{\partial\gamma}\omega_{1,n}(X,\gamma)\right|_{\gamma=\gamma_0}\right]'\cdot \left(\widehat\gamma - \gamma_0\right) + \mathbb{E}_n\left[R_{1,h}(X,\widehat\gamma,\gamma_0)\right]
\\
=&
\left(\mathbb{E}_n-\mathbb{E}\right)\left[\left.\frac{\partial}{\partial\gamma}\omega_{1,n}(X,\gamma)\right|_{\gamma=\gamma_0}\right]' \cdot \left(\widehat\gamma - \gamma_0\right)
\\
&+
\mathbb{E}\left[\left.\frac{\partial}{\partial\gamma}\omega_{1,n}(X,\gamma)\right|_{\gamma=\gamma_0}\right]' \cdot \left(\widehat\gamma - \gamma_0\right) 
+ \mathbb{E}_n\left[R_{1,h}(X,\widehat\gamma,\gamma_0)\right]
\\
=&
\mathbb{E}\left[\left.\frac{\partial}{\partial\gamma}\omega_{1,n}(X,\gamma)\right|_{\gamma=\gamma_0}\right]' \cdot \mathbb{E}_n[\phi]
+
o_p\left(n^{-1/2}\right)
\\
=&
\mathbb{E}_n \left[ \mathbb{E}\left[\left.\frac{\partial}{\partial\gamma}\omega_{1,n}(X,\gamma)\right|_{\gamma=\gamma_0}\right]' \cdot \phi \right]
+ o_p\left(n^{-1/2}\right)
\end{align*}
where the second equality is due to (\ref{eq:ex:taylor1}), and
the second to last equality is due to (\ref{eq:lemma:ex:trimmed_mean1}), (\ref{eq:lemma:ex:trimmed_mean2}), and (\ref{eq:lemma:ex:trimmed_mean3}).
\end{proof}

\subsection{Bias Estimator}
\begin{lemma}\label{lemma:ex:bias_estimator}
If Assumptions \ref{a:ex:m} (iii)--(iv), \ref{a:ex:s} (i)--(iv), \ref{a:ex:h} (ii), and \ref{a:ex:gamma} are satisfied, then
\begin{align*}
&\mathbb{E}_n\left[g_A(X,\widehat\gamma)^{\kappa-1} \cdot \left( S\left(\frac{g_A(X,\widehat\gamma)}{h_n}\right) - 1\right)\right]
-
\mathbb{E}_n\left[g_A(X,\gamma_0)^{\kappa-1} \cdot \left( S\left(\frac{g_A(X,\gamma_0)}{h_n}\right) - 1 \right)\right]
\\
=&
\mathbb{E}_n \left[  \mathbb{E}\left[ \left.\frac{\partial}{\partial\gamma'} \omega_{2,\kappa,n}(X,\gamma)\right|_{\gamma=\gamma_0} \right] \cdot \phi \right]
+ o_p\left(n^{-1/2}\right)
\end{align*}
for each $\kappa \in \{1,...,k-1\}$.
\end{lemma}
\begin{proof}
We present four sets of auxiliary calculations.
First,
\begin{align}
\omega_{2,\kappa,n}(x,\gamma) - \omega_{2,\kappa,n}(x,\gamma_0)
&=
\left.\frac{\partial}{\partial\gamma'} \omega_{2,\kappa,n}(X,\gamma)\right|_{\gamma=\gamma_0} \left(\gamma - \gamma_0\right)
+
R_{2,\kappa,h}(x,\gamma,\gamma_0)
\qquad\text{where}
\label{eq:ex:taylor2}
\\
\abs{R_{2,\kappa,h}(x,\gamma,\gamma_0)}
& \leq
\sup_{\gamma\in\Gamma} \norm{\frac{\partial^2}{\partial\gamma\partial\gamma'} \omega_{2,\kappa,n}(x,\gamma)} \cdot \norm{\gamma-\gamma_0}^2
\label{eq:ex:taylor2_remainder}
\end{align}
by Taylor's theorem under Assumptions \ref{a:ex:m} (iii) and \ref{a:ex:s} (ii).
Second,
\begin{align}
\abs{\mathbb{E}_n\left[R_{2,\kappa,h}(X,\widehat\gamma,\gamma_0)\right]}
\leq
\mathbb{E}_n\left[ \sup_{\gamma\in\Gamma} \norm{\frac{\partial^2}{\partial\gamma\partial\gamma'} \omega_{2,\kappa,n}(X,\gamma)} \right] \cdot \norm{\widehat\gamma - \gamma_0}^2
=
o_p\left(n^{-1/2}\right),
\label{eq:lemma:ex:bias_estimator1}
\end{align}
where the inequality is due to (\ref{eq:ex:taylor2_remainder}) 
and 
the equality is due to Assumption \ref{a:ex:gamma} and (\ref{eq:ex:omega2_2}) of Lemma \ref{lemma:ex:omega2} under Assumptions \ref{a:ex:m} (iii)--(iv), \ref{a:ex:s} (i)--(iv), and \ref{a:ex:h} (ii).
Third,
\begin{align}
\mathbb{E}\left[ \left.\frac{\partial}{\partial\gamma'} \omega_{2,\kappa,n}(X,\gamma)\right|_{\gamma=\gamma_0} \right] \cdot \left(\widehat\gamma - \gamma_0\right) 
=
\mathbb{E}\left[ \left.\frac{\partial}{\partial\gamma'} \omega_{2,\kappa,n}(X,\gamma)\right|_{\gamma=\gamma_0} \right] \cdot \mathbb{E}_n[\phi]
+
o_p\left(n^{-1/2}\right),
\label{eq:lemma:ex:bias_estimator2}
\end{align}
where the equality is due to Assumption \ref{a:ex:gamma} and (\ref{eq:ex:omega2_1}) of Lemma \ref{lemma:ex:omega2} under Assumptions \ref{a:ex:m} (iii)--(iv), \ref{a:ex:s} (i)--(iv), and \ref{a:ex:h} (ii).
Fourth,
\begin{align}
\left(\mathbb{E}_n-\mathbb{E}\right)\left[ \left.\frac{\partial}{\partial\gamma'} \omega_{2,\kappa,n}(X,\gamma)\right|_{\gamma=\gamma_0} \right] \cdot \left(\widehat\gamma - \gamma_0\right)
=
o_p\left(n^{-1/2}\right),
\label{eq:lemma:ex:bias_estimator3}
\end{align}
where the equality is due to Assumption \ref{a:ex:gamma} and (\ref{eq:ex:omega2_1}) of Lemma \ref{lemma:ex:omega2} under Assumptions \ref{a:ex:m} (iii)--(iv), \ref{a:ex:s} (i)--(iii), and \ref{a:ex:h} (ii).

We now obtain
\begin{align*}
&\mathbb{E}_n\left[g_A(X,\widehat\gamma)^{\kappa-1} \cdot \left( S\left(\frac{g_A(X,\widehat\gamma)}{h_n}\right) - 1\right)\right]
-
\mathbb{E}_n\left[g_A(X,\gamma_0)^{\kappa-1} \cdot \left( S\left(\frac{g_A(X,\gamma_0)}{h_n}\right) - 1 \right)\right]
\\
=&
\mathbb{E}_n\left[ \omega_{2,\kappa,n}(X,\widehat\gamma) -  \omega_{2,\kappa,n}(X,\gamma_0)\right]
\\
=&
\mathbb{E}_n\left[ \left.\frac{\partial}{\partial\gamma'} \omega_{2,\kappa,n}(X,\gamma)\right|_{\gamma=\gamma_0} \right] \cdot \left(\widehat\gamma - \gamma_0\right) + \mathbb{E}_n\left[ R_{2,\kappa,h}(X,\widehat\gamma,\gamma_0)\right]
\\
=&
\left(\mathbb{E}_n-\mathbb{E}\right)\left[ \left.\frac{\partial}{\partial\gamma'} \omega_{2,\kappa,n}(X,\gamma)\right|_{\gamma=\gamma_0} \right] \cdot \left(\widehat\gamma - \gamma_0\right)
\\
&+
\mathbb{E}\left[ \left.\frac{\partial}{\partial\gamma'} \omega_{2,\kappa,n}(X,\gamma)\right|_{\gamma=\gamma_0} \right] \cdot \left(\widehat\gamma - \gamma_0\right) 
+ \mathbb{E}_n\left[ R_{2,\kappa,h}(X,\widehat\gamma,\gamma_0)\right]
\\
=&
\mathbb{E}\left[ \left.\frac{\partial}{\partial\gamma'} \omega_{2,\kappa,n}(X,\gamma)\right|_{\gamma=\gamma_0} \right] \cdot \mathbb{E}_n[\phi]
+
o_p\left(n^{-1/2}\right)
\\
=&
\mathbb{E}_n \left[\mathbb{E}\left[ \left.\frac{\partial}{\partial\gamma'} \omega_{2,\kappa,n}(X,\gamma)\right|_{\gamma=\gamma_0} \right] \cdot \phi \right]
+ o_p\left(n^{-1/2}\right)
\end{align*}
where the second equality is due to (\ref{eq:ex:taylor2}), and
the second to last equality is due to (\ref{eq:lemma:ex:bias_estimator1}), (\ref{eq:lemma:ex:bias_estimator2}), and (\ref{eq:lemma:ex:bias_estimator3}).
\end{proof}

\subsection{Sieve Estimator}
\begin{lemma}\label{lemma:ex:sieve_estimator}
If Assumptions \ref{a:ex:psi} and \ref{a:ex:gamma} are satisfied, then
\begin{align*}
\widehat m^{(\kappa)}(0;\widehat\gamma) - m^{(\kappa)}(0)
=
\mathbb{E}_n\left[ \psi_\kappa + \left.\frac{\partial}{\partial\gamma'} m^{(\kappa)}(0;\gamma)\right|_{\gamma=\gamma_0} \cdot \phi \right] 
+ o_p\left(n^{-1/2} h_n^{1-\kappa}\right)
\end{align*}
for each $\kappa \in \{1,...,k-1\}$.
\end{lemma}
\begin{proof}
We present a couple of auxiliary calculations.
First,
\begin{align}
\widehat m^{(\kappa)}(0;\gamma) - \widehat m^{(\kappa)}(0;\gamma_0)
&=
\left.\frac{\partial}{\partial\gamma'} \widehat m^{(\kappa)}(0;\gamma)\right|_{\gamma=\gamma_0}(\gamma - \gamma_0) + R_{3,\kappa}(\gamma,\gamma_0)
\qquad\text{where}
\label{eq:ex:taylor3}
\\
\abs{R_{3,\kappa}(\gamma,\gamma_0)}
&\leq
\sup_{\gamma\in\Gamma} \norm{\frac{\partial^2}{\partial\gamma\partial\gamma'} \widehat m^{(\kappa)}(0;\gamma)}
\cdot \norm{\gamma-\gamma_0}^2
\label{eq:ex:taylor3_remainder}
\end{align}
by Taylor's theorem under Assumption \ref{a:ex:psi} (iii).
Second,
\begin{align}
\abs{R_{3,\kappa}(\gamma,\gamma_0)}
&\leq
\sup_{\gamma\in\Gamma} \norm{\frac{\partial^2}{\partial\gamma\partial\gamma'} \widehat m^{(\kappa)}(0;\gamma)}
\cdot \norm{\gamma-\gamma_0}^2
=
o_p\left(n^{-1/2}\right),
\label{eq:lemma:ex:sieve_estimator1}
\end{align}
where the equality is due to Assumptions \ref{a:ex:psi} (iii) and \ref{a:ex:gamma}.

We now obtain
\begin{align*}
\widehat m^{(\kappa)}(0;\widehat\gamma) - m^{(\kappa)}(0)
&=
\widehat m^{(\kappa)}(0;\gamma_0) - m^{(\kappa)}(0) +
\widehat m^{(\kappa)}(0;\widehat\gamma) - \widehat m^{(\kappa)}(0;\gamma_0)
\\
&= \left(\mathbb{E}_n-\mathbb{E}\right)\left[\psi\right] 
+
\left.\frac{\partial}{\partial\gamma'} \widehat m^{(\kappa)}(0;\gamma)\right|_{\gamma=\gamma_0}(\widehat\gamma - \gamma_0)
+
o_p\left(n^{-1/2} h_n^{1-\kappa}\right)
\\
&= \left(\mathbb{E}_n-\mathbb{E}\right)\left[\psi\right] 
+
\left.\frac{\partial}{\partial\gamma'} m^{(\kappa)}(0;\gamma)\right|_{\gamma=\gamma_0} \cdot \left(\mathbb{E}_n-\mathbb{E}\right) \left[\phi\right]
+
o_p\left(n^{-1/2} h_n^{1-\kappa}\right)
\end{align*}
where 
the second equality uses Assumption \ref{a:ex:psi} (i)--(ii) and (\ref{eq:ex:taylor3})--(\ref{eq:lemma:ex:sieve_estimator1}),
and the third equality is due to Assumptions \ref{a:ex:psi} (iv) and \ref{a:ex:gamma}.
\end{proof}

\subsection{Linear Representation}\label{sec:ex:linear_representation}

\begin{lemma}[Linear Representation]\label{lemma:ex:influence_function}
If Assumptions \ref{a:ex:m} (iii)--(v), \ref{a:ex:s},  \ref{a:ex:h} (ii), \ref{a:ex:psi}, and \ref{a:ex:gamma} are satisfied, then
$$
\widehat\theta(h_n)-\theta(h_n)-\lambda(h_n)
=
(\mathbb{E}_n-\mathbb{E})[Z(h_n)]+o_p(n^{-1/2}).
$$
\end{lemma}
\begin{proof}
Lemma \ref{lemma:ex:trimmed_mean} shows that
\begin{align}
\mathbb{E}_n\left[\omega_{1,n}\left(X,\widehat\gamma\right)\right]
-
\mathbb{E}_n\left[ \omega_{1,n}\left(X,\gamma_0\right)\right]
-
\mathbb{E}_n \left[ \mathbb{E}\left[\left. \frac{\partial}{\partial\gamma'}\omega_{1,n}(X,\gamma)\right|_{\gamma=\gamma_0} \right] \cdot \phi \right]
=
o_p\left(n^{-1/2}\right)
\label{eq:lemma:ex:influence_function1}
\end{align}
under Assumptions \ref{a:ex:m} (iii) and (v), \ref{a:ex:s} (ii)--(iii), \ref{a:ex:h} (ii), and \ref{a:ex:gamma}.
Lemma \ref{lemma:ex:bias_estimator} shows that
\begin{align}
&\mathbb{E}_n\left[\omega_{2,\kappa,n}\left(X,\widehat\gamma\right)\right]
-
\mathbb{E}_n\left[\omega_{2,\kappa,n}\left(X,\gamma_0\right)\right]
-
\mathbb{E}_n \left[  \mathbb{E}\left[ \left.\frac{\partial}{\partial\gamma'} \omega_{2,\kappa,n}(X,\gamma)\right|_{\gamma=\gamma_0} \right] \cdot \phi \right]
= 
o_p\left(n^{-1/2}\right)
\label{eq:lemma:ex:influence_function2}
\end{align}
for each $\kappa \in \{1,...,k-1\}$
under Assumptions \ref{a:ex:m} (iii)--(iv), \ref{a:ex:s} (i)--(iv), \ref{a:ex:h} (ii), and \ref{a:ex:gamma}.
Lemma \ref{lemma:ex:sieve_estimator} shows that
\begin{align}
\widehat m^{(\kappa)}(0;\widehat\gamma) - m^{(\kappa)}(0)
-
\mathbb{E}_n\left[ \psi_\kappa \right] 
- 
\left.\frac{\partial}{\partial\gamma'} m^{(\kappa)}(0;\gamma)\right|_{\gamma=\gamma_0} \cdot \mathbb{E}_n \left[ \phi \right] 
= 
o_p\left(n^{-1/2} h_n^{1-\kappa}\right)
\label{eq:lemma:ex:influence_function3}
\end{align}
for each $\kappa \in \{1,...,k-1\}$
under Assumptions \ref{a:ex:psi} and \ref{a:ex:gamma}.
Furthermore,
\begin{align}
\mathbb{E}\left[\omega_{2,\kappa,n}\left(X,\gamma_0\right)\right] = O\left(h_n^{\kappa-1}\right)
\label{eq:lemma:ex:influence_function4}
\end{align}
by Assumption \ref{a:ex:s} (i).
It follows from (\ref{eq:lemma:ex:influence_function1}), (\ref{eq:lemma:ex:influence_function2}), (\ref{eq:lemma:ex:influence_function3}), and (\ref{eq:lemma:ex:influence_function4}) that
\begin{align*}
&
\widehat\theta(h_n)-\lambda(h_n)-\mathbb{E}_n[Z(h_n)]\\
=&
\mathbb{E}_n\left[\omega_{1,n}(X,\widehat\gamma)\right]-\mathbb{E}_n\left[\omega_{1,n}(X,\gamma_0)\right]\\&-\mathbb{E}_n \left[ \mathbb{E}\left[\left. \frac{\partial}{\partial\gamma'} \omega_{1,n}(X,\gamma)\right|_{\gamma=\gamma_0} \right] \cdot \phi \right]\\
&+
\sum_{\kappa=1}^{k-1} \frac{\mathbb{E}_n\left[\omega_{2,\kappa,n}(X,\widehat\gamma)-\omega_{2,\kappa,n}(X,\gamma_0)-\mathbb{E}\left[\left.\frac{\partial}{\partial\gamma'}\omega_{2,\kappa,n}(X,\gamma)\right|_{\gamma=\gamma_0}\right]\phi\right]}{\kappa !} \cdot  m^{(\kappa)}(0)
\\
&+
\sum_{\kappa=1}^{k-1} \frac{\mathbb{E}\left[\omega_{2,\kappa,n}(X,\gamma_0)\right]}{\kappa !} \cdot (\widehat m^{(\kappa)}(0;\widehat\gamma)-m^{(\kappa)}(0)-\mathbb{E}_n[\psi_\kappa]-\left.\frac{\partial}{\partial\gamma'} m^{(\kappa)}(0;\gamma)\right|_{\gamma=\gamma_0} \cdot \mathbb{E}_n[\phi])
\\
&+
\sum_{\kappa=1}^{k-1} \frac{\mathbb{E}_n\left[\omega_{2,\kappa,n}(X,\widehat\gamma)-\omega_{2,\kappa,n}(X,\gamma_0)\right]+(\mathbb{E}_n-\mathbb{E})\left[\omega_{2,\kappa,n}(X,\gamma_0)\right]}{\kappa !} \cdot (\widehat m^{(\kappa)}(0;\widehat\gamma)-m^{(\kappa)}(0))
\\
=& 
o_p\left(n^{-1/2}\right)
\end{align*}
under
Assumptions \ref{a:ex:m} (iii)--(v), \ref{a:ex:s} (i)--(iii),  \ref{a:ex:h} (ii), \ref{a:ex:psi}, and \ref{a:ex:gamma}.
Therefore, we obtain
$$
\widehat\theta(h_n)-\theta(h_n)-\lambda(h_n)
=
\widehat\theta(h_n)-\mathbb{E}\left[Z(h_n)\right]-\lambda(h_n)
=
(\mathbb{E}_n-\mathbb{E})[Z(h_n)]+o_p(n^{-1/2})
$$
as claimed.
\end{proof}

\bibliography{mybib}

\end{document}